%% file: sec0main.tex
\documentclass[11pt]{article}
\usepackage[prependcaption, colorinlistoftodos,textsize=small]{todonotes}
\usepackage{geometry,setspace,appendix,verbatim}
\usepackage{authblk}
\usepackage{graphicx}
\usepackage{subcaption}  
\usepackage{caption}
\usepackage{amsmath,amsfonts,amssymb,amsthm,bbm}
\allowdisplaybreaks
\usepackage{enumerate}
\usepackage{comment}

\usepackage{natbib}
\usepackage{hyperref}
\allowdisplaybreaks

\newcommand{\indep}{\perp \!\!\! \perp} 

\newcommand{\littleOp}{o_p}

\newcommand{\Norm}{{\bf{\rm N}}\,} 

\newcommand{\iidsim}{\stackrel{{\mathrm{iid}}}{\sim}}

\newcommand{\norm}[1]{\left\lVert #1 \right\rVert}

\newcommand{\rbr}[1]{\left(#1\right)}
\newcommand{\sbr}[1]{\left[#1\right]}
\newcommand{\cbr}[1]{\left\{#1\right\}}

\newcommand{\indist}{\xrightarrow[]{d}}
\newcommand{\inprob}{\xrightarrow[]{p}}

\def\*#1{\mathbf{#1}} 

\def\bQ{{\mathbb Q}}

\def\bR{{\mathbb R}}
\def\bE{{\mathbb E}}

\def\sP{{\mathcal P}}
\def\sQ{{\mathcal Q}}
\def\sR{{\mathcal R}}

\def\sX{{\mathcal X}}
\def\sY{{\mathcal Y}}

\def\Var{\mathrm{Var}}
\def\Cov{\mathrm{Cov}}

\newtheorem{theorem}{Theorem}[section]
\newtheorem{corollary}{Corollary}[subsection]

\newtheorem{proposition}[theorem]{Proposition}
\newtheorem{remark}{Remark}[subsection]
\renewcommand{\theremark}{%
  \ifnum\value{subsection}=0
    \thesection.\arabic{remark}
  \else
    \thesubsection.\arabic{remark}
  \fi
}
\makeatother

\newtheorem{assumption}{Assumption}[section]

\newcommand{\labest}[1]{\hat{#1}_{n}}

\newcommand{\allfest}[1]{\hat{#1}_{n+N}^f}
\newcommand{\labfest}[1]{\hat{#1}_{n}^{f}}

\newcommand{\unlabfest}[1]{\hat{#1}_{N}^{f}}

\newcommand{\target}[1]{{#1}_{0}}
\newcommand{\appipwtarget}[1]{{#1}_{0,\text{IPW}}}
\newcommand{\aucfulltarget}[1]{{#1}_{0}^{\text{full},\text{AUC}}}

\newcommand{\appiowtarget}[1]{{#1}_{0,\text{IOW}}}
\newcommand{\targetf}[1]{{#1}^{f}_{0}}
\newcommand{\opt}[1]{{#1}^{*}}
\newcommand{\ppest}[1]{\hat{#1}^{\text{PP}}}
\newcommand{\labmdest}[1]{\hat{#1}_{\text{lab}}}
\newcommand{\fulllabmdest}[1]{\hat{#1}_{\labfull}}
\newcommand{\unlabeledlabmdest}[1]{\hat{#1}_{\labunlabeled}}
\newcommand{\labeledlabmdest}[1]{\hat{#1}_{\lablabeled}}
\newcommand{\labmdestpi}[1]{\hat{#1}_{\text{lab}}^{\hat{\pi}}}
\newcommand{\fulllabmdestpi}[1]{\hat{#1}_{\frac{C}{\hat{\pi}(X)}}}

\newcommand{\labmdestf}[1]{\hat{#1}^f_{\text{lab}}}
\newcommand{\fulllabmdestf}[1]{\hat{#1}^f_{\labfull}}
\newcommand{\unlabeledlabmdestf}[1]{\hat{#1}^f_{\labunlabeled}}
\newcommand{\labeledlabmdestf}[1]{\hat{#1}^f_{\lablabeled}}
\newcommand{\labmdestpif}[1]{\hat{#1}^{f,\hat{\pi}}_{\text{lab}}}
\newcommand{\fulllabmdestpif}[1]{\hat{#1}^{f}_{\frac{C}{\hat{\pi}(X)}}}

\newcommand{\fullmdestf}[1]{\hat{#1}^{f}}
\newcommand{\unlabeledfullmdestf}[1]{\hat{#1}^{f}_\fullunlabeled}
\newcommand{\labeledfullmdestf}[1]{\hat{#1}^{f}_\fulllabeled}
\newcommand{\mdppest}[1]{\hat{#1}_{\text{MD}}^{\text{PP}}}

\newcommand{\fullest}[1]{\hat{#1}}
\newcommand{\labeledest}[1]{\hat{#1}^{\labeled}}
\newcommand{\unlabeledest}[1]{\hat{#1}^{\unlabeled}}

\newcommand{\recest}[2]{\hat{#1}_{\hat{\omega}}^{\hat{#2}}}
\newcommand{\recomega}[2]{\hat{#1}_{\omega}^{\hat{#2}}}
\newcommand{\Lab}{\mathcal{L}}
\newcommand{\Unlab}{\mathcal{U}}

\def\joint{{\mathbb{P}^*}}
\def\jointcov{{\mathbb{Q}^*}}
\def\jointf{{\mathbb{P}^*_{X,f(X)}}}
\def\marginal{{\mathbb{P}^*_{X}}}
\def\marginalalt{{\tilde{\mathbb{P}}^*_{X}}}
\def\conditional{{\mathbb{P}^*_{Y\mid X}}}

\def\jointapp{{\mathbb{P}_{R,f(X),Y}^*}}
\def\marginalapp{{\mathbb{P}^*_{R,f(X)}}}
\def\conditionalapp{{\mathbb{P}^*_{Y\mid R,f(X)}}}

\def\IFtheta{{\phi}}
\def\IFdelta{{\psi}}
\def\TPR{{\text{TPR}(\alpha)}}
\def\FPR{{\text{FPR}(\alpha)}}
\def\AUC{{\text{AUC}}}

\def\unlabeled{{\text{IOW}}}
\def\full{{\text{full}}}
\def\labeled{{C=1}}
\def\lab{{\text{lab}}}

\def\unlabeled{{C=0}}

\def\labfull{{\frac{C}{\pi(X)}}}
\def\labunlabeled{{q(X)C}}
\def\fullunlabeled{{1-\pi(X)}}
\def\lablabeled{{C}}
\def\fulllabeled{{\pi(X)}}

\begin{document}
\title{Generalized Prediction-Powered Inference, \\
with Application to Binary Classifier Evaluation }

\doublespacing

\author[1]{Runjia Zou}
\author[1,2]{Daniela Witten}
\author[1,3,4]{Brian Williamson}

\affil[1]{Department of Biostatistics, University of Washington, WA, USA}
\affil[2]{Department of Statistics, University of Washington, WA, USA}
\affil[3]{Biostatistics Division, Kaiser Permanente Washington Health Research Institute, WA, USA}
\affil[4]{Vaccine and Infectious Disease Division, Fred Hutchinson Cancer Center, WA, USA}

\date{}

\maketitle

\begin{abstract}
In the partially-observed outcome setting, a recent set of proposals known as ``prediction-powered inference" (PPI) involve (i) applying a pre-trained machine learning model to predict the response, and then (ii) using these predictions to obtain an estimator of the parameter of interest with  asymptotic variance  no greater than that which would be obtained using only the labeled observations. While existing PPI proposals consider estimators arising from  M-estimation, in this paper we generalize PPI to any  regular asymptotically linear estimator. Furthermore, by situating PPI within the context of an existing rich literature on missing data and semi-parametric efficiency theory, 
we show that while PPI does not achieve the semi-parametric efficiency lower bound outside of very restrictive and unrealistic scenarios, it can be viewed as a computationally-simple alternative to proposals in that literature. We exploit connections to that literature to propose modified PPI estimators that can handle three distinct forms of covariate distribution shift. Finally, we illustrate these developments by constructing PPI estimators of true positive rate, false positive rate, and area under the curve via numerical studies.
\end{abstract}

\newpage
\input{sec1intro}
\input{sec2generalizedPPI}
\input{sec3covshift}

\input{sec4estpi}

\input{sec5binarymeric}

\input{sec6simulations}
\input{sec7realdata}

\input{sec8discussion}

\section*{Acknowledgments}
We thank R. Yates Coley for helpful conversations that contributed to the development of this work. 
This work was funded by NIH R01MH125821 and NIH R37AI131771 to BDW, and by 
ONR N00014-23-1-2589, NSF DMS
2322920 and  2514344, and
NIH 5P30DA048736 to DW.
The content is solely the responsibility of the authors and does not necessarily represent the official views of the funding agencies.
\bibliography{ref}
\appendix 
\input{secAppendix}

\end{document}

%% file: sec1intro.tex
\section{Introduction}
\label{sec:intro}
Consider random variables $Y \in \sY$ and $X \in \sX$, where $(X,Y)\sim \joint$. The goal is to make inference on a functional of $\joint$ in a setting
where we have access to $n$ independent observations of $(X,Y) \sim \joint = \marginal\conditional $, and an additional $N$ independent observations of $X \sim \marginal$, where $n \ll N$. We will also consider the setting of \textit{covariate distribution shift}, where the additional $N$ observations follow a different marginal distribution.
 Some examples are as follows:

\begin{list}{}{}
\item{\bf Example 1:} 
Social media companies have access to the text of  billions of user comments ($X$) per day. However, a definitive label ($Y$) for whether or not a comment is problematic   requires  time-consuming and expensive hand-labeling  by a human. 
 \item{\bf Example 2:}
 Obtaining clinical measurements ($X$) for a patient --- for instance, blood pressure, heart rate, etc. --- is often relatively low-cost. However, obtaining accurate clinical labels ($Y$) for each patient --- for instance, disease diagnosis --- may be much more expensive. 
 \end{list}
Of course, we could simply conduct inference on  $\target{\theta} := \theta(\joint)$ using the $n$ observations of $(X,Y)$. However, we wish to use the additional $N$ observations of $X$  to improve inference on $\target{\theta}$ (i.e., to reduce the variance).  Broadly speaking, this problem has been well-studied through the lens of missing data, semi-parametric efficiency, causal inference, and survey sampling  \citep{robins1994estimation,robins1995semiparametric,tsiatis2006semiparametric,bickel1993efficient,chernozhukov2018double,deville1992calibration,breslow2009improved}. 

Recently, a set of proposals in the machine learning literature, collectively referred to as \emph{prediction-powered inference} (PPI), have centered on the setting where we  have access to a function $f(\cdot): \sX \rightarrow \sY$ (e.g., a prediction model), such that $Y \approx f(X)$; this function is either deterministic or estimated based on an additional independent dataset (so that we can view it as deterministic conditional on that other data).  In either case, we assume that the operating characteristics of $f(\cdot)$ are unknown. For instance, in the context of Example 1, $f(\cdot)$ might be a large language model that predicts whether a given comment is offensive; 
and in the context of Example 2, it might be a decision tree that predicts disease risk from clinical covariates.
We define $\Lab_n:=\cbr{(X_i, f(X_i), Y_i)}_{i=1}^n$ to be the ``labeled" dataset, and $\, \Unlab_N := \cbr{(X_i, f(X_i))}_{i=n+1}^{n+N} $ to be the ``unlabeled" dataset. The key idea of PPI is to augment an estimator constructed only using $\Lab_n$ with a rectifier that incorporates both $f$ and $\Unlab_N$.

In the original PPI paper, \cite{angelopoulos2023prediction}  proposed an estimator that incorporates both $\Lab_n$ and  $\Unlab_N$ into the estimation of $\theta(\joint)$ in the context of M-estimation. Extensive follow-up work in the PPI setting has led to improved M-estimators or Z-estimators \citep{angelopoulos2023ppi++, zrnic2024cross,gan2024prediction,miao2025assumption}. 

Despite this recent interest in PPI, substantial gaps remain: 
\begin{enumerate}
    \item \emph{Existing PPI estimators are defined in terms of M-estimators or Z-estimators; this limits the application of PPI.} 
    While \cite{miao2024task} and \cite{gronsbell2026another} provide preliminary steps towards extending PPI beyond M- and Z- estimation, their proposals lack theoretical grounding and applicability beyond least squares, respectively.  
    \item \emph{Existing PPI estimators do not provide a general and theoretically-justified strategy for handling the possibility of covariate distribution shift between $\Lab_n$ and $\Unlab_N$.}  
   
 While methods in \cite{gronsbell2026another} and \cite{zrnic2024active} can address this problem in certain limited scenarios, they do not provide a theoretically-justified set of solutions to target the various parameters that may be of interest in the context of  covariate distribution shift. 
\end{enumerate}
The main contributions of this paper are as follows:

\begin{itemize} 
\item We develop a general framework for PPI, in which the target of inference $\theta(\joint)$ need not arise as the solution to a population-level M-estimation  problem. Given a regular asymptotically linear estimator (ALE) of $\theta(\joint)$  based on $\Lab_n$, we can devise another ALE of  $\theta(\joint)$ based on $\Lab_n$ and $\Unlab_N$ that is guaranteed to have variance no higher than the original estimator. We also present the analytic form of the variance. 

\item The existing PPI literature does not clarify the motivation for the form of the rectifier used. We provide novel theoretical results justifying the form of the rectifier used in PPI, and we clarify that the PPI estimator will not achieve the semiparametric efficiency bound \citep{robins1994estimation,robins1995semiparametric} except under strong conditions.

\item We connect our proposed PPI framework to the rich and well-developed literature of missing-data methods  \citep{robins1994estimation,robins1995semiparametric,chernozhukov2018double}. Specifically, we show that PPI estimators can be view as part of the augmented inverse probability weighted estimator class. While PPI estimators will in general not achieve the semi-parametric efficiency lower bound in the missing data setting, they have the advantage of straightforward computation for a variety of targets of inference. This complements recent work by \cite{xu2025unified}, who studied efficient semi-parametric estimators in the semi-supervised (rather than the missing data) setting. 

\item Inspired by the missing data framework, we propose PPI estimators that can accommodate   covariate distribution  shift. We incorporate both $\Lab_n$ and $\Unlab_N$ to conduct inference for  three distinct targets that may arise in the context of distribution shift. Two of these targets, functionals only of the labeled or unlabeled data,  have not been previously considered in the  PPI literature, and proposals for the third (a functional of the combined labeled and unlabeled data) lack theoretical guarantees achieved by our estimator.

\item We illustrate our ideas in the context of binary classifier evaluation. Specifically, we propose PPI estimators of the  true positive rate (TPR), false positive rate (FPR), and area under the curve (AUC); these three quantities are not typically formulated through the lens of Z- or M-estimation, and thus could not be easily handled using the prior PPI literature. 
\end{itemize}

The rest of this paper is organized as follows. In Section~\ref{sec:method}, we show that any ALE that can be computed on $\Lab_n$ can be easily extended to obtain a PPI estimator that incorporates $\Unlab_N$; further,  we connect this proposal to  the missing data literature. In Sections~\ref{sec:covariate-shift} and \ref{sec:estpi}, we exploit this connection to propose PPI estimators that are suitable in the context of covariate distribution shift, for three distinct targets of inference. 
In Section~\ref{sec:binaryclassifier}, we apply these ideas to develop PPI estimators of TPR, FPR, and AUC. We study the performance of our estimators on  simulated data in Section~\ref{sec:simulation} and apply them to wine quality data \citep{cortez2009modeling} in Section~\ref{sec:realdata}. We provide concluding remarks in Section~\ref{sec:discussion}. All theoretical results are proven in the Appendix. Scripts for replicating our numerical results are available at \url{https://github.com/rz001010/generalizedPPI}.

%% file: sec2generalizedPPI.tex
\section{A generalized PPI estimator} \label{sec:method}

\subsection{A rectified regular ALE}

Our interest lies in a functional $\theta: \sP \rightarrow \bR$, with target $\target{\theta} := \theta(\joint)$, where $\joint \in \sP$. Assume that $(X_1,Y_1), \ldots, (X_n, Y_n)  \iidsim
 \marginal\conditional$, and that $X_{n+1},\ldots,X_{n+N} \iidsim
 \marginal$.  Let $\labest{\theta}:=\hat{\theta}\rbr{(X_i,Y_i)_{i=1}^n}$ denote   a regular ALE of $\target{\theta}$   with influence function $\IFtheta(x,y)$. That is,
\begin{equation}
\labest{\theta} - \target{\theta} = \frac{1}{n} \sum_{i=1}^n \IFtheta{(X_i, Y_i)} + \littleOp(n^{-1/2}).
\label{eq:labestALE}
\end{equation} 
This estimator only uses the $n$ labeled observations. To incorporate the unlabeled observations $X_{n+1},\ldots,X_{n+N}$, we first  define $\target{\delta}:=\delta\rbr{\marginal}$, for which $\labest{\delta}:=\hat{\delta}\rbr{(X_i)_{i=1}^n}$  and  $\hat{\delta}_{n+N}:=\hat{\delta}\rbr{(X_i)_{i=1}^{n+N}}$ are regular ALEs with influence function $\psi(x)$,  so that  
\begin{equation}
\labest{\delta} - \target{\delta} = \frac{1}{n} \sum_{i=1}^n \IFdelta(X_i) + o_p(n^{-1/2}) \;\;\;  \text{and} \;\;\; \hat{\delta}_{n+N} - \target{\delta} = \frac{1}{n+N} \sum_{i=1}^{n+N} \IFdelta(X_i) + o_p(n^{-1/2}).
\label{eq:recestALE}
\end{equation}
We then  define the  \textit{rectified} estimator
\begin{equation}
\recest{\theta}{\delta}:= \labest{\theta}+\hat{\omega}\rbr{\hat{\delta}_{n+N}-\labest{\delta}}, \label{eq:recest}
\end{equation}
where we will refer to $\hat{\delta}_{n+N}-\labest{\delta}$ as the \emph{rectifier}. 
The following result guarantees that if $\hat\omega$ is chosen carefully, then this estimator will improve upon (and never worsen) the asymptotic variance of $\labest{\theta}$. 

\begin{theorem}[Improving $\hat\theta_n$ with a rectifier] 
\label{thm:recest} 
 Suppose that $\hat{\omega}\xrightarrow{p}\omega$ and  $\frac{n}{N}\rightarrow\lambda$. 
The following hold:
\begin{enumerate}[(i)]
    \item Define \begin{equation}\label{eq:recest_var}
    \rbr{\sigma_{\omega}^{\IFdelta}}^2: =\bE[\IFtheta(X,Y)^2]-\frac{2\omega}{1+\lambda}\bE[\IFtheta(X,Y)\IFdelta(X)]+\frac{\omega^2}{1+\lambda}\bE[\IFdelta(X)^2].\end{equation}
    Then 
    \begin{equation*}\sqrt{n}\rbr{\recest{\theta}{\delta}-\target{\theta}}\indist \Norm\rbr{0, \rbr{\sigma_{\omega}^{\IFdelta}}^2}.
    \end{equation*}   
    \item  Define
    $\opt{\omega}:= \frac{\bE[\IFtheta(X,Y)\IFdelta(X)]}{\bE[\IFdelta(X)^2]}$. When viewed as a function of $\omega$, $(\sigma_{\omega}^{\IFdelta})^2$ achieves its minimal value of $$  (\sigma_{\opt{\omega}}^{\IFdelta})^2=\bE[\IFtheta(X,Y)^2]-\frac{\rbr{\bE[\IFtheta(X,Y)\psi(X)]}^2}{(1+\lambda)\bE[\psi(X)^2]}$$  when  $\omega=\opt{\omega}$. 
This is  no greater than $\bE[\IFtheta(X,Y)^2]$, the asymptotic variance of $\hat{\theta}_n$.
    \item Define $\opt{\IFdelta}(X):=\bE[\IFtheta(X,Y)\mid X]$.  When viewed as a function of $\psi(\cdot)$, $(\sigma_{\omega^*}^{\IFdelta})^2$ achieves its minimal value of $$(\sigma_{\opt{\omega}}^{\opt{\IFdelta}})^2:=\bE[\IFtheta(X,Y)^2]-\frac{1}{1+\lambda}\Var(\bE[\IFtheta(X,Y)\mid X])$$ when   
 $\opt{\IFdelta}(X)=\bE[\IFtheta(X,Y)\mid X]$. 
\end{enumerate}
\end{theorem}
We prove Theorem \ref{thm:recest} in  Appendix \ref{sec:recestproof}.
Parts (i) and (ii) of Theorem~\ref{thm:recest} reveal that $\recest{\theta}{\delta}$ is consistent for $\theta_0$ and that the asymptotic variance of $\recest{\theta}{\delta}$ will be no higher than that of $\hat\theta_n$, for \emph{any} regular ALEs $\hat\delta_n$ and $\hat\delta_{n+N}$ in \eqref{eq:recestALE}, provided that $\hat\omega \xrightarrow{p} \omega^*$. 

 Furthermore, part (iii) reveals that the asymptotic variance of $\recest{\theta}{\delta}$ is minimized when the regular ALE $\hat\delta_n$ has influence function $\IFdelta(X)=\bE[\IFtheta(X,Y)\mid X]$.  This suggests that in constructing \eqref{eq:recest}, we may wish for
 $\hat\delta_n$ and $\hat\delta_{n+N}$ to have  influence function $\bE[\IFtheta(X,Y) \mid X ]$. Indeed,  this strategy has deep roots in the missing data literature \citep{robins1994estimation,robins1995semiparametric,chernozhukov2018double} and is also considered in  \cite{xu2025unified} in the context of semi-supervised learning; see Section~\ref{sec:missingdata} for further discussion. 
 
 However, motivated by the PPI literature, in this paper we take a different strategy for constructing $\hat\delta_n$ and $\hat\delta_{n+N}$. As we will see, though generally it will not achieve the semi-parametric efficiency bound  \citep[unlike, for instance,][]{robins1994estimation,robins1995semiparametric,xu2025unified,testa2025semiparametric}, our strategy has advantages from a practical and computational perspective. 
 
 \subsection{A rectifer motivated by PPI} \label{sec:rectifier}
 
 %
 %
%
%
%
Motivated by the PPI proposals of \cite{angelopoulos2023prediction} and \cite{angelopoulos2023ppi++}, 
 suppose that we have access to a pre-trained prediction model $f(\cdot): \mathcal{X} \rightarrow \mathcal{Y}$. 
 %
Recalling the definition of  $\labest{\theta}$ in \eqref{eq:labestALE}, we define 
\begin{equation}
\labfest{\theta} := \hat\theta\rbr{(X_i,f(X_i))_{i=1}^n}\;\;\; \text{and} \;\;\; \allfest{\theta} := \hat\theta\rbr{(X_i,f(X_i))_{i=1}^{n+N}}. 
\label{eq:theta-f}
\end{equation}
Noting that $f(\cdot)$ is independent of the  data, $\Lab_n$ and $\Unlab_N$, 
 it follows that $\labfest{\theta}$ and $\allfest{\theta}$ are regular ALEs of $\targetf{\theta}:=\theta(\jointf)$  with influence function $x\rightarrow \IFtheta(x,f(x))$, provided that $\targetf{\theta} < \infty$ is well-defined and  there is no first-order degeneracy in $\labfest{\theta}-\targetf{\theta}$ and $\allfest{\theta}-\targetf{\theta}$. 

We now revisit the estimator \eqref{eq:recest}, this time using $\hat{\theta}^f_{n+N} - \labfest{\theta}$ as the rectifier, i.e.,
\begin{equation} 
\ppest{\theta}:=\labest{\theta}+\hat{\omega}\rbr{\allfest{\theta}-\labfest{\theta}}.
\label{eq:PPI}
\end{equation} 
In the existing PPI literature (see, e.g., \cite{angelopoulos2023prediction,angelopoulos2023ppi++}), the rectifier involves $\unlabfest{\theta}$ instead of $\allfest{\theta}$;  here we use $\allfest{\theta}$ in anticipation of Section 3. However, the theoretical results in this section also hold when using $\unlabfest{\theta}$ as long as $\hat{\omega}$ is defined appropriately (Appendix \ref{sec:onlyN}).

\begin{proposition}[Asymptotically linear PPI]\label{prop:PPest}

Suppose that  
  $\frac{n}{N}\rightarrow\lambda$ and that
 \begin{equation}\label{eq:omegaf}
\hat\omega \inprob    \omega:= \frac{\bE[\IFtheta(X,Y)\IFtheta(X,f(X))]}{\bE[\IFtheta(X,f(X))^2]}.
\end{equation}
Then
    \begin{equation}
    \sqrt{n}\rbr{\ppest{\theta}-\target{\theta}}\indist \Norm \rbr{0,\bE[\IFtheta(X,Y)^2]-\frac{\rbr{\bE[\IFtheta(X,Y)\IFtheta(X,f(X))]}^2}{(1+\lambda)\bE[\IFtheta(X,f(X))^2]}}.
    \label{eq:PPI-asymptotic}
    \end{equation}
\end{proposition}
Thus, the asymptotic variance of $\ppest{\theta}$ is less than that of  $\hat\theta_n$ by the amount $\frac{\rbr{\bE[\IFtheta(X,Y)\IFtheta(X,f(X))]}^2}{(1+\lambda)\bE[\IFtheta(X,f(X))^2]}$. 


We know from part (iii) of Theorem~\ref{thm:recest} that a rectifier with influence function $\psi(X) = \bE[\phi(X,Y) \mid X]$ achieves the minimal asymptotic variance out of estimators of the form \eqref{eq:recest} (with $\hat\omega$ as specified in part (ii) of Theorem~\ref{thm:recest}).  It is straightforward to see that $\phi(X, f(X))=\bE[\phi(X,Y) \mid X]$, and thus the estimator $\ppest{\theta}$ achieves this minimal variance, when both of the following two conditions hold: 
\begin{assumption}(Oracle situation for PPI)\label{assmp:oracle_case}

    \begin{enumerate}
\item The influence function $\IFtheta(X,Y)$ is conditionally linear in the second term, that is $\bE\sbr{\IFtheta(X,Y)\mid X}=\IFtheta\rbr{X,\bE\sbr{Y\mid X}}$, and  
\item $f(X)=\bE\sbr{Y \mid X}$.
\end{enumerate}
\end{assumption}

The first of these conditions is satisfied for many estimators of interest: examples include standard estimators of the  population mean, coefficients in a linear regression,  and  coefficients in a generalized linear model \citep{xu2025unified}, as well as  TPR, FPR, and AUC  considered in Section~\ref{sec:binaryclassifier}. However, the second of these conditions requires that the machine learning model $f(\cdot)$ perfectly capture the true mean $\bE[Y \mid X]$.  Because this assumption will almost never hold in practice, \emph{the PPI estimator $\ppest{\theta}$ will generally not achieve the optimal variance guaranteed by Theorem~\ref{thm:recest}}. Nonetheless, Proposition \ref{prop:PPest} guarantees that $\ppest{\theta}$ is always consistent for the target $\target{\theta}$ and has asymptotic variance no greater than $\labest{\theta}$.

We close this section with a comment on the ways in which our generalized PPI estimator \eqref{eq:PPI} differs from the prior PPI literature:
\begin{itemize}
\item Unlike the majority of prior work that provides theoretical guarantees, we do not require that the target arise from a population-level M-estimation problem. 
\item While the approach of \cite{miao2024task} does not center on an M-estimator, it relies on the bootstrap for variance estimation and computation of $\hat\omega$ in \eqref{eq:PPI}. Our approach employs an analytic variance calculation, substantially decreasing computation time.
\end{itemize}

\subsection{Connections to the missing data framework}\label{sec:missingdata}
 

We now 
re-formulate some of the ideas in the previous sections through the lens of missing data. We 
introduce $C_i$, a binary indicator variable that equals $1$ if the response $Y_i$ is observed, and $0$ otherwise. We further assume that $C\indep (Y, X)$; that is, the setting is missing completely at random. 
Then, $\Lab_n \cup \Unlab_N$ is of the form $\rbr{X_i, C_i, C_iY_i}_{i=1}^{n+N}$. While Theorem \ref{thm:recest} and Proposition \ref{prop:PPest} assumed that $\frac{n}{N}\rightarrow \lambda$, in what follows we will instead assume that $\pi:=\Pr(C=1)=\frac{\lambda}{1+\lambda}$.

For a general target  $\target{\theta}$, defining $\hat\pi:=\frac{1}{n+N}\sum_{i=1}^{n+N} C_i$, and following the logic in Section~\ref{sec:rectifier}, for a properly defined weight $\hat{\omega}_{\text{MD}}$, a PPI estimator can be constructed as
\begin{align}
    \mdppest{\theta}-\target{\theta}=\frac{1}{n+N}\sum_{i=1}^{n+N}\rbr{\frac{C_i}{\hat{\pi}}\IFtheta(X_i,Y_i)+\hat{\omega}_{\text{MD}} \left(1-\frac{C_i}{\hat{\pi}}\right)\IFtheta(X_i,f(X_i))},\label{eq:aipw-f}
\end{align} where the subscript ``MD'' stands for ``missing data''.
The next result (proven in Appendix \ref{sec:proofprop2.3}) establishes that  $\mdppest{\theta} \overset{d}{=}  \ppest{\theta}$ where $\ppest{\theta}$ was defined in \eqref{eq:PPI}.

\begin{proposition}
    \label{prop:aipw-f}
    Recall the definition of  $\omega$ in \eqref{eq:omegaf}. 
Suppose that $\hat\pi=\frac{1}{n+N}\sum_{i=1}^{n+N} C_i$, $\hat{\omega}_{\text{MD}}\xrightarrow{p}\omega$, and $\hat\omega\xrightarrow{p}\omega$. Then, $\mdppest{\theta}\overset{d}{=} \ppest{\theta}$. 
\end{proposition}
The PPI estimator in \eqref{eq:aipw-f} belongs to the class of general augmented inverse propensity weighted (AIPW) estimators. It is natural to wonder how  $\mdppest{\theta}$  connects to the well-studied efficient AIPW estimator 
$\hat{\theta}_{\text{AIPW}}$ \citep{robins1994estimation,robins1995semiparametric,tsiatis2006semiparametric}, which can be written as 
\begin{align}
    \hat{\theta}_{\text{AIPW}}-\target{\theta}=\frac{1}{n+N}\sum_{i=1}^{n+N}\rbr{\frac{C_i}{\hat{\pi}}\IFtheta(X_i,Y_i)+\left(1-\frac{C_i}{\hat{\pi}}\right)\widehat{\bE}[\IFtheta(X_i,Y_i)|X_i]},\label{eq:aipw}
\end{align}
 provided  that $\phi(x,y)$ is the \emph{efficient} influence function and 
 $\hat{\bE}[\phi(X,Y) \mid X = x]$ is a ``good enough'' estimator (in the sense of  \cite{robins1994estimation,robins1995semiparametric,tsiatis2006semiparametric}) of $\bE[\phi(X,Y)\mid X = x]$. Then, the estimator in \eqref{eq:aipw}
 achieves the semiparametric efficiency bound.

 In general, \eqref{eq:aipw-f} will be less efficient than  \eqref{eq:aipw}, unless the (unrealistic) Assumption \ref{assmp:oracle_case} is satisfied. Indeed, the connection between $\ppest{\theta}$ and $\mdppest{\theta}$ has been noted by \cite{gronsbell2026another} and \cite{zrnic2024active} in the special cases of estimating linear regression coefficients and  M-estimation. Proposition~\ref{prop:aipw-f} formalizes and generalizes this connection.

Knowing that the estimator $\ppest{\theta}$ will not achieve the optimal variance outside of the (unrealistic) Assumption ~\ref{assmp:oracle_case}, and that there exist estimators along the lines of \eqref{eq:aipw} that asymptotically achieve the optimal variance under certain assumptions \citep{robins1994estimation,xu2025unified,testa2025semiparametric}, why would we use $\ppest{\theta}$? The reasons are as follows: 
\begin{itemize}
\item For the estimators proposed by \cite{robins1994estimation,xu2025unified,testa2025semiparametric} to be efficient, they require, for any estimand under consideration: (i) a derivation of the efficient influence function, and (ii) a good enough $\hat{\bE}[\phi(X,Y) \mid X = x]$, where $\IFtheta(x,y)$ is the \emph{efficient} influence function. By contrast, the PPI approach begins with an (at-hand) estimator, and its rectifier is easy to compute. It does not require deriving the efficient influence function or solving an estimation equation,  puts no assumptions on the function $f$ that we use to construct the rectifier in \eqref{eq:PPI}, and does not require estimation of $\bE[\phi(X,Y) \mid X = x]$. 
\item Despite the asymptotic guarantees of 
\cite{robins1994estimation,xu2025unified,testa2025semiparametric}, these estimators are not guaranteed to outperform  $\ppest{\theta}$ in finite samples. 
%
%
\item Both the mathematical derivations underlying $\ppest{\theta}$ \eqref{eq:PPI} and its construction are quite straightforward and accessible to a non-specialist. 
\end{itemize}

The connection  detailed in Proposition~\ref{prop:aipw-f} between $\ppest{\theta}$ in \eqref{eq:PPI} and the missing data framework suggests that the machinery for covariate distribution shift that has been developed in the missing data setting \citep{robins1994estimation,chen2000unified,westreich2017transportability,dahabreh2020extending}  could be used to extend $\ppest{\theta}$. This is the topic of Section~\ref{sec:covariate-shift}. 

%% file: sec3covshift.tex
\section{PPI under covariate distribution shift}\label{sec:covariate-shift}
We have so far assumed that the distribution of $X$ is the same for the labeled and unlabeled data: in other words,  $X_1,\ldots,X_{n+N} \iidsim \marginal$. We now consider the possibility of \textit{covariate distribution shift} \citep{horvitz1952generalization,robins1994estimation,angelopoulos2023prediction}. 

 \cite{xu2025unified} point out that using $\Lab_n \cup \Unlab_N$ for inference, as opposed to just using the labeled data $\Lab_n$, is beneficial only when the target parameter depends on the distribution of $X$: that is, when $\theta\rbr{\joint}=\theta\rbr{\marginal\conditional}$ is not equal to 
$\theta\rbr{\marginalalt\conditional}$ for $\marginal \neq \marginalalt$. Thus, in the setting where PPI has the potential to improve upon inference using $\Lab_n$, covariate distribution shift will affect the target parameter, and therefore needs to be adequately addressed.

 In the previous section, our target of inference was $\theta_0 = \theta(\joint)$.  Later in this section, we will apply the functional delta method, which requires Hadamard differentiability. Therefore, to set ourselves up for this notationally, from this point forward we will assume that  the target can be written as $\theta_0 = \Phi(\target{F}^*)$, where the functional $\Phi$ maps from the space of cumulative distribution functions (CDFs) to the real line $\mathbb{R}$ and is Hadamard differentiable with respect to supremum norm \citep{shao2008mathematical}, and where $\target{F}^*$ is the CDF of $(X,Y)$ relative to the population of interest. 

We will now expand upon the notation from Section~\ref{sec:missingdata}. Our data consist of  $\rbr{X_i,C_i,C_iY_i}_{i=1}^{n+N}$, where $C_i$ is an indicator variable that equals $1$ if the  $i$th observation  is labeled; the ideal (but unobserved) full data is  $\rbr{X_i,C_i,Y_i}_{i=1}^{n+N}\iidsim \jointcov \in \sQ$.  
In the setting of covariate distribution shift, there are multiple possible targets of interest: 
 \begin{enumerate} \item The  parameter with respect to the full data distribution $\jointcov_{X,Y}$ (both labeled and unlabeled data), with corresponding CDF $\target{F}(x,y)=\jointcov_{X,Y}(X\leq x,Y\leq y)$, which we will refer to as 
 \begin{equation}\label{eq:targetfull}
     \target{\theta}:=\Phi(\target{F}).
 \end{equation} 
 \item The  parameter with respect to the distribution of the unlabeled data $\jointcov_{X,Y\mid C=0}$, with corresponding CDF $\target{F}^\unlabeled(x,y)=\jointcov_{X,Y\mid C=0}(X\leq x,Y\leq y)$, which we will refer to as \begin{equation}\label{eq:targetunlab}
     \target{\theta}^{\unlabeled}:=\Phi(\target{F}^\unlabeled).
 \end{equation}
 \item  The  parameter with respect to the distribution of the labeled data $\jointcov_{X,Y\mid C=1}$, with corresponding CDF $\target{F}^\labeled(x,y)=\jointcov_{X,Y\mid C=1}(X\leq x,Y\leq y)$, which we will refer to as 
 \begin{equation} \label{eq:targetlab}
     \target{\theta}^{\labeled}:=\Phi(\target{F}^\labeled).
 \end{equation}
  \end{enumerate}
 We will consider these three targets in Sections~\ref{sec:covdistall}, \ref{sec:covdistunlab}, and \ref{sec:covdistlab}, respectively. In practice, any of these targets can be of scientific interest, as we illustrate using our two examples from Section~\ref{sec:intro}. 
\begin{list}{}{}
\item{\bf Example 1, continued:} A social media company may hand-label user comments  only for a (designed) selected subset of accounts (e.g., accounts with many followers). Thus, the distribution of $X$ between the labeled and unlabeled data may differ.
\begin{itemize}
    \item The overall proportion of problematic comments on the platform is $\target{\theta}$. 
    \item The proportion of problematic content among comments that were never labeled is  $\target{\theta}^{\unlabeled}$. 
    \item The proportion  of problematic content among comments that have been labeled  is  $\target{\theta}^{\labeled}$. 
\end{itemize}
We hope to incorporate the $X$ values for the non-labeled comments to improve inference. 
\item{\bf Example 2, continued:}  Suppose 
$Y$ corresponds to a rare disease outcome.
Then,  measuring $Y$ on a small simple random sample from the total population   is often not efficient.  Instead, one could measure $Y$ on a (designed) non-uniform  sample, where sampling is based on  $X$. The target can be the mean of $Y$, e.g., the incidence rate of a rare disease. 
\begin{itemize}
    \item  The incidence rate in the total population is $\target{\theta}$.
    \item 
    The incidence rate among people without $Y$ measured is $\target{\theta}^{\unlabeled}$. 
    \item  The incidence rate among people with $Y$ measured is $\target{\theta}^{\labeled}$. 
\end{itemize}
\end{list}
We now introduce some additional notation and assumptions. 
Define $\pi(X)=\Pr(C=1\mid X)$. We assume the following:
\begin{assumption}\label{assmp:covdistshift}
(Assumptions for PPI under covariate distribution shift)
\begin{itemize}
    \item Independent and identically distributed: $\rbr{X_i,C_i,Y_i}_{i=1}^{n+N}\iidsim \jointcov \in \sQ$.
    \item Conditional independence: $C\indep Y\mid X$.
    \item Overlap: $\epsilon < \pi(X) < 1 - \epsilon$ for some $\epsilon > 0$.
\end{itemize}
\end{assumption} 
We assume that $\pi(X)$ is known: e.g., the missingness is determined by the study design. In Section \ref{sec:estpi} we discuss the case where $\pi(X)$ needs to be estimated.

By construction, the distribution of $(X,Y)$ in $\Lab_n$ is $\jointcov_{X,Y\mid C=1}$, while the distribution of $(X,Y)$ in $\Unlab_N$ is $\jointcov_{X,Y\mid C=0}$. As in Section~\ref{sec:method}, we suppose that we have access to a deterministic function $f$ such that $f(X) \approx Y$, and suppose $\target{F}^{f}(x,y):=\jointcov_{X,f(X)}(X\leq x,f(X)\leq y)$, $\target{F}^{\unlabeled,f}(x,y):=\jointcov_{X,f(X)\mid C=0}(X\leq x,f(X)\leq y)$, and $\target{F}^{\labeled,f}(x,y):=\jointcov_{X,f(X)\mid C=1}(X\leq x,f(X)\leq y)$ are well-defined CDFs.

Recall that $\target{F}^*$ is the CDF of
$(X, Y)$ in the population of interest, i.e., it equals either $\target{F}$, $\target{F}^\unlabeled$, or $\target{F}^\labeled$. Given  $m$ i.i.d.  draws from the population of interest and $F_{m}(x,y):=\frac{1}{m}\sum_{i=1}^{m}I(X_i\leq x,Y_i\leq y)$, it follows that $\Phi(F_{m})$ is a regular ALE of the target, $\Phi(F_{0}^*)$, with influence function  \begin{equation} \label{eq:generalif_fdm}
    \phi(x,y):=\dot{\Phi}(\target{F}^*;\delta_{x,y}-\target{F}^*),
\end{equation} where  $\delta_{x,y}(u_1,u_2)= I(x\leq u_1,y\leq u_2)$, and $\dot{\Phi}$ is the Gateaux derivative of $\Phi$. 

In Sections \ref{sec:covdistall}, \ref{sec:covdistunlab}, and \ref{sec:covdistlab},  we will use this insight to obtain PPI estimators of $\Phi(\target{F})$, $\Phi(\target{F}^\unlabeled)$, and $\Phi(\target{F}^\labeled)$.
 Results for regular ALEs that do not take the special form given in Sections \ref{sec:covdistall}, \ref{sec:covdistunlab}, and \ref{sec:covdistlab} can be found in Section \ref{sec:covdistgeneral}. 

\subsection{PPI for targets of the full data distribution}\label{sec:covdistall}

To accommodate covariate distribution shift when the target is $\target{\theta}$ in \eqref{eq:targetfull}, we apply inverse probability weighting (IPW) \citep{horvitz1952generalization,hejazi2021efficient}. 
That is, we inverse-weight the labeled observations by the probability of being labeled, and combine this with predictions on the unlabeled observations. 
Define \begin{align*}
     &F_\labfull(x,y):=\frac{\sum_{i=1}^{n+N}\frac{C_i}{\pi(X_i)}I(X_i\leq x,Y_i\leq y)}{\sum_{i=1}^{n+N}\frac{C_i}{\pi(X_i)}} \text{ and } \fulllabmdest{\theta}:=\Phi\rbr{F_\labfull}. \\
       &F^{f}_\labfull(x,y):=\frac{\sum_{i=1}^{n+N}\frac{C_i}{\pi(X_i)}I(X_i\leq x,f(X_i)\leq y)}{\sum_{i=1}^{n+N}\frac{C_i}{\pi(X_i)}} \text{ and }  \fulllabmdestf{\theta}:=\Phi\rbr{F^{f}_\labfull}. \\
       &F^{f}(x,y):=\frac{1}{n+N}\sum_{i=1}^{n+N}I(X_i\leq x,f(X_i)\leq y) \text{ and } \fullmdestf{\theta}:=\Phi\rbr{F^{f}}.  
\end{align*}

Similar to the estimator proposed in \eqref{eq:PPI}, we construct the PPI estimator as
\begin{equation}\label{eq:estfull}
    \fullest{\theta}:= \fulllabmdest{\theta}+\hat{\omega}\rbr{\fullmdestf{\theta}-\fulllabmdestf{\theta}}.
\end{equation}

\begin{remark}\label{rem:full_fdm}
    $\fulllabmdest{\theta}$ is a regular ALE of $\target{\theta}$ \eqref{eq:targetfull} with influence function evaluated at a single observation  $\varphi_\lab(X_i,Y_i,C_i):=\frac{C_i}{\pi(X_i)}\IFtheta(X_i,Y_i)$; $\fulllabmdestf{\theta}$ is a regular ALE with influence function evaluated at a single observation $\varphi_\lab(X_i,f(X_i),C_i):=\frac{C_i}{\pi(X_i)}\IFtheta(X_i,f(X_i))$; and $\fullmdestf{\theta}$ is a regular ALE  with influence function evaluated at a single observation $\varphi(X_i,f(X_i),C_i):=\IFtheta(X_i,f(X_i))$, where $\IFtheta$ is defined in \eqref{eq:generalif_fdm}. These claims are established in Appendix \ref{app:fdm-ale}.
\end{remark}

Our next result establishes the asymptotic behavior of $\fullest{\theta}$ \eqref{eq:estfull}. 
\begin{theorem}[PPI for targets of the full data distribution]\label{thm:ppiipw}
Under Assumption \ref{assmp:covdistshift} and the further assumption  that $\hat{\omega}\inprob \omega$, where
     \begin{equation}
     \label{eq:omega_full}
         \omega:=\frac{\Cov\left(\frac{C_i}{\pi(X_i)}\phi(X_i,Y_i),\rbr{\frac{C_i}{\pi(X_i)}-1}\phi(X_i,f(X_i))\right)}{\Var\left(\rbr{1-\frac{C_i}{\pi(X_i)}}\phi(X_i,f(X_i))\right)},
     \end{equation} it follows that 
$$\sqrt{n+N}\rbr{\fullest{\theta}-\target{\theta}}\indist N(0,\sigma^2),$$ where $\sigma^2=\Var\left(\frac{C_i}{\pi(X_i)}\phi(X_i,Y_i)\right)-\frac{\Cov\left(\frac{C_i}{\pi(X_i)}\phi(X_i,Y_i),\rbr{1-\frac{C_i}{\pi(X_i)}}\phi(X_i,f(X_i))\right)^2}{\Var\left(\rbr{1-\frac{C_i}{\pi(X_i)}}\phi(X_i,f(X_i))\right)}$.
\end{theorem}
\begin{remark}
The forms of $\omega$ in \eqref{eq:omega_full} and the asymptotic variance decrease connect closely with the influence functions $\varphi_\lab(X_i,Y_i,C_i)$, $\varphi_\lab(X_i,f(X_i),C_i)$, and $\varphi(X_i,f(X_i),C_i)$ in Remark \ref{rem:full_fdm}, in the sense that  
    \begin{equation}\label{eq:omega_if}
        \omega=\frac{\Cov\Big(\varphi_\lab(X_i,Y_i,C_i),\varphi_\lab(X_i,f(X_i),C_i)-\varphi(X_i,f(X_i),C_i)\Big)}{\Var\Big(\varphi(X_i,f(X_i),C_i)-\varphi_\lab(X_i,f(X_i),C_i)\Big)},
    \end{equation}
and the asymptotic variance of $\fullest{\theta}$ \eqref{eq:estfull} improves on that of  $\fulllabmdest{\theta}$ by an amount
\begin{equation}\label{eq:var_improvement_if}
    \frac{\Cov\Big(\varphi_\lab(X_i,Y_i,C_i),\varphi_\lab(X_i,f(X_i),C_i)-\varphi(X_i,f(X_i),C_i)\Big)^2}{\Var\Big(\varphi(X_i,f(X_i),C_i)-\varphi_\lab(X_i,f(X_i),C_i)\Big)},
\end{equation}
where the asymptotic variance of $\fulllabmdest{\theta}$  follows from its influence function in Remark \ref{rem:full_fdm}.
\end{remark}

When Assumption \ref{assmp:oracle_case} is satisfied ---  that is, when
 $\phi(x,y)$ is linear in its second term, and
   $f(X)=\bE[Y\mid X]$ --- 
it follows from Theorem~\ref{thm:ppiipw} that 
\begin{equation}\label{eq:oracle_full}
    \sigma^2=\Var\left(\frac{C_i}{\pi(X_i)}\phi(X_i,Y_i)\right)-\bE\sbr{\rbr{\frac{1}{\pi(X_i)}-1}\bE[\phi(X_i,Y_i)\mid X_i]^2}, 
\end{equation}
which aligns with the semi-parametric efficiency bound \citep{robins1994estimation,van2000asymptotic,tsiatis2006semiparametric} provided that $\phi(X,Y)$ is the efficient influence function. Of course, as noted earlier, Assumption~\ref{assmp:oracle_case} is highly unrealistic. 

\cite{zrnic2024active} consider estimation of $\target{\theta}$ in the context of M-estimation: however, their estimator omits the  $\hat{\omega}$ term in \eqref{eq:estfull} and consequently may have asymptotic variance  greater than that of $\fulllabmdest{\theta}$. 
For regular ALEs, \cite{kluger2025prediction} provide an alternative estimator of $\target{\theta}$ that, like ours, is guaranteed to have asymptotic variance no greater than that of $\fulllabmdest{\theta}$; however, the first term inside their rectifier only uses the unlabeled data, i.e., in 
\eqref{eq:estfull} they replace $\hat\theta_f$ with an estimator involving only the unlabeled observations. Consequently, even under Assumption~\ref{assmp:oracle_case}, their estimator does not achieve the semi-parametric efficiency bound. (A straightforward proof can be found in Section~\ref{sec:ptd_comparison}.) 

By contrast,  \cite{testa2025semiparametric}  take an entirely different approach: rather than incorporating $f(X)$ into the estimator as in \eqref{eq:estfull}, they instead derive  the efficient influence function, and acquire a ``good enough'' estimator of its conditional mean, and therefore achieve the semi-parametric efficiency bound in \eqref{eq:oracle_full}. However, estimating the conditional mean of an influence function can be challenging when labeled data are limited, and so in finite samples our PPI-motivated approach may be preferable.  

\subsection{PPI for targets of the unlabeled data distribution}\label{sec:covdistunlab}
To accommodate covariate distribution shift when the target is $\target{\theta}^{\unlabeled}$ \eqref{eq:targetunlab}, we apply inverse odds weighting (IOW) \citep{westreich2017transportability,dahabreh2020extending}, which involves inverse weighting the labeled observations by the odds of being labeled. 

  We define the inverse odds of being labeled as  $q(X):=\frac{\Pr(C=0\mid X)}{\Pr(C=1\mid X)}$.
Further, define
\begin{align*}
    &F^\unlabeled_\labunlabeled(x,y):=\frac{\sum_{i=1}^{n+N}q(X_i)C_iI(X_i\leq x,Y_i\leq y)}{\sum_{i=1}^{n+N}q(X_i)C_i},\; \unlabeledlabmdest{\theta}:=\Phi\rbr{F^\unlabeled_\labunlabeled},\\
    &F^{\unlabeled,f}_\labunlabeled(x,y):=\frac{\sum_{i=1}^{n+N}q(X_i)C_iI(X_i\leq x,f(X_i)\leq y)}{\sum_{i=1}^{n+N}q(X_i)C_i},\; \unlabeledlabmdestf{\theta}:=\Phi\rbr{F^{\unlabeled,f}_\labunlabeled},\\
    &F^{\unlabeled,f}_\fullunlabeled(x,y):=\frac{\sum_{i=1}^{n+N}(1-\pi(X_i))I(X_i\leq x,f(X_i)\leq y)}{\sum_{i=1}^{n+N}(1-\pi(X_i))},\; \unlabeledfullmdestf{\theta}:=\Phi\rbr{F^{\unlabeled,f}_\fullunlabeled}.
\end{align*}

We propose our PPI estimator for $\target{\theta}^{\unlabeled}$ \eqref{eq:targetunlab} as
\begin{equation}\label{eq:unlabest}
    \unlabeledest{\theta}:=\unlabeledlabmdest{\theta}+\hat{\omega}^\unlabeled\rbr{\unlabeledfullmdestf{\theta}-\unlabeledlabmdestf{\theta}}.
\end{equation}

\begin{remark}\label{rem:unlab_fdm}
    $\unlabeledlabmdest{\theta}$ is a regular ALE of $\target{\theta}^{\unlabeled}$ \eqref{eq:targetunlab} with the influence function evaluated at a single observation  $\varphi_\lab(X_i,Y_i,C_i):=\frac{q(X_i)C_i}{\Pr(C=0)}\IFtheta(X_i,Y_i)$; $\unlabeledlabmdestf{\theta}$ is a regular ALE with the influence function evaluated at a single observation $\varphi_\lab(X_i,f(X_i),C_i):=\frac{q(X_i)C_i}{\Pr(C=0)}\IFtheta(X_i,f(X_i))$; and $\unlabeledfullmdestf{\theta}$ is a regular ALE  with the influence function evaluated at a single observation  $\varphi(X_i,f(X_i),C_i):=\frac{1-\pi(X_i)}{\Pr(C=0)}\IFtheta(X_i,f(X_i))$, where $\IFtheta$ is defined in \eqref{eq:generalif_fdm}. 
\end{remark}
Our next result states the asymptotic behaviour of $\unlabeledest{\theta}$ \eqref{eq:unlabest}. 


\begin{theorem}[PPI for targets of the unlabeled data distribution]\label{thm:ppiiow} Under Assumption \ref{assmp:covdistshift} and the further assumption that $\hat{\omega}^\unlabeled\inprob \omega^\unlabeled$, where 
    \begin{equation}\label{eq:omega_unlab}
        \omega^\unlabeled:=\frac{\Cov\Big(q(X_i)C_i\phi(X_i,Y_i),\sbr{q(X_i)C_i-(1-\pi(X_i))}\phi(X_i,f(X_i))\Big)}{\Var\Big(\sbr{1-\pi(X_i)-q(X_i)C_i}\phi(X_i,f(X_i))\Big)},
    \end{equation}
    it follows that
    $$\sqrt{n+N}\rbr{\unlabeledest{\theta}-\target{\theta}^\unlabeled}\xrightarrow{d}N\rbr{0,\sigma_\unlabeled^2},$$ where $\sigma_\unlabeled^2=\frac{\Var\rbr{q(X_i)C_i\IFtheta(X_i,Y_i)}}{\Pr(C=0)^2}-\frac{\Cov\Big(q(X_i)C_i\phi(X_i,Y_i),\;\sbr{1-\pi(X_i)-q(X_i)C_i}\phi(X_i,f(X_i))\Big)^2}{\Pr(C=0)^2\Var\Big(\sbr{1-\pi(X_i)-q(X_i)C_i}\phi(X_i,f(X_i))\Big)}$.
\end{theorem}
We note that  plugging  $\varphi_\lab(X_i,Y_i,C_i)$, $\varphi_\lab(X_i,f(X_i),C_i)$, and $\varphi(X_i,f(X_i),C_i)$  defined in Remark \ref{rem:unlab_fdm} into \eqref{eq:omega_if}  yields \eqref{eq:omega_unlab}.  Similarly, plugging these influence functions into \eqref{eq:var_improvement_if} yields the asymptotic variance improvement  (relative to $\unlabeledlabmdest{\theta}$) seen in Theorem \ref{thm:ppiiow}.

$\unlabeledest{\theta}$ \eqref{eq:unlabest} enables one to utilize the available data on Y from a non-target population to improve inference on the target population.

\subsection{PPI for targets of the labeled data distribution}\label{sec:covdistlab}

We now consider the target $\target{\theta}^\labeled$ \eqref{eq:targetlab}. Of course, this target can be recovered using the labeled data alone. However, as in Section \ref{sec:method}, we can utilize the unlabeled data to reduce variance. 
Define 
\begin{align*}
&F^\labeled_\lablabeled(x,y):=\frac{\sum_{i=1}^{n+N}C_iI(X_i\leq x,Y_i\leq y)}{\sum_{i=1}^{n+N}C_i},\; \labeledlabmdest{\theta}:=\Phi\rbr{F^\labeled_\lablabeled}.\\
&F^{\labeled,f}_\lablabeled(x,y):=\frac{\sum_{i=1}^{n+N}C_iI(X_i\leq x,f(X_i)\leq y)}{\sum_{i=1}^{n+N}C_i},\; \labeledlabmdestf{\theta}:=\Phi\rbr{F^{\labeled,f}_\lablabeled}.\\
&F^{\labeled,f}_\fulllabeled(x,y):=\frac{\sum_{i=1}^{n+N}\pi(X_i)I(X_i\leq x,f(X_i)\leq y)}{\sum_{i=1}^{n+N}\pi(X_i)},\; \labeledfullmdestf{\theta}:=\Phi\rbr{F^{\labeled,f}_\fulllabeled}.
\end{align*}

We propose a PPI estimator for $\target{\theta}^{\labeled}$ in \eqref{eq:targetlab} as
\begin{equation}\label{eq:labest}
    \labeledest{\theta}:=\labeledlabmdest{\theta}+\hat{\omega}^\labeled\rbr{\labeledfullmdestf{\theta}-\labeledlabmdestf{\theta}}.
\end{equation}

\begin{remark}\label{rem:lab_fdm}
    $\labeledlabmdest{\theta}$ is a regular ALE of $\target{\theta}^{\labeled}$ \eqref{eq:targetlab} with the influence function evaluated at a single observation  $\varphi_\lab(X_i,Y_i,C_i):=\frac{C_i}{\Pr(C=1)}\IFtheta(X_i,Y_i)$; $\labeledlabmdestf{\theta}$ is a regular ALE with the influence function evaluated at a single observation $\varphi_\lab(X_i,f(X_i),C_i):=\frac{C_i}{\Pr(C=1)}\IFtheta(X_i,f(X_i))$; and $\labeledfullmdestf{\theta}$ is a regular ALE  with the influence function evaluated at a single observation  $\varphi(X_i,f(X_i),C_i):=\frac{\pi(X_i)}{\Pr(C=1)}\IFtheta(X_i,f(X_i))$. 
\end{remark}

Our next result establishes the asymptotic behavior of 
$\labeledest{\theta}$ \eqref{eq:labest}.



\begin{theorem}[PPI for targets of the labeled data distribution]\label{thm:ppiow} Under Assumption \ref{assmp:covdistshift} and the further assumption that  $\hat{\omega}^\labeled\inprob \omega^\labeled$, where
   \begin{equation}\label{eq:omega_lab}
       \omega^\labeled:=\frac{\Cov\Big(C_i\phi(X_i,Y_i),\sbr{C_i-\pi(X_i)}\phi(X_i,f(X_i))\Big)}{\Var\Big(\sbr{\pi(X_i)-C_i}\phi(X_i,f(X_i))\Big)},
   \end{equation} it follows that 
    $$\sqrt{n+N}\rbr{\labeledest{\theta}-\target{\theta}^\labeled}\xrightarrow{d}N\rbr{0,\sigma_\labeled^2},$$ where $\sigma_\labeled^2=\frac{\Var\rbr{C_i\IFtheta(X_i,Y_i)}}{\Pr(C=1)^2}-\frac{\Cov\Big(C_i\phi(X_i,Y_i),\;\sbr{\pi(X_i)-C_i}\phi(X_i,f(X_i))\Big)^2}{\Pr(C=1)^2\Var\Big(\sbr{\pi(X_i)-C_i}\phi(X_i,f(X_i))\Big)}$.
\end{theorem}
Again,  plugging  $\varphi_\lab(X_i,Y_i,C_i)$, $\varphi_\lab(X_i,f(X_i),C_i)$, and $\varphi(X_i,f(X_i),C_i)$  defined in Remark \ref{rem:lab_fdm} into \eqref{eq:omega_if}  yields \eqref{eq:omega_lab}.  Similarly, plugging these influence functions into \eqref{eq:var_improvement_if} yields the asymptotic variance improvement  (relative to $\labeledlabmdest{\theta}$) seen in Theorem \ref{thm:ppiow}.

%% file: sec4estpi.tex
\section{Covariate shift with unknown labeling mechanism}\label{sec:estpi}

We assumed in Section~\ref{sec:covariate-shift} that the labeling mechanism, $\pi(X)$, is known. 
However, this may not be the case in practice. For instance, in the context of Example 2,  individuals living in high-income areas with better access to healthcare resources may be more likely to be labeled because they are better able to attend medical appointments. In this setting, our interest may lie in $\target{\theta}$ in \eqref{eq:targetfull}, a summary of the full population distribution.



In this section, we extend the ideas in Section~\ref{sec:covariate-shift} to allow for unknown $\pi(X)$. Specifically, we estimate this quantity from  $(C_i,X_i,C_iY_i)_{i=1}^{n+N}$. Here, we only consider the target $\target{\theta}$ \eqref{eq:targetfull} from Section \ref{sec:covdistall}. Estimators for the targets $\target{\theta}^\unlabeled$ \eqref{eq:targetunlab} and $\target{\theta}^\labeled$ \eqref{eq:targetlab} can be derived analogously.

 By construction, $\fulllabmdest{\theta}$ from  Section \ref{sec:covdistall} can be viewed as a functional of $\bQ_{n+N}$ and $\pi$, where $\bQ_{n+N}$ is the empirical law of $(C_i,X_i, Y_i)_{i=1}^{n+N}$. Therefore, we can write $\fulllabmdest{\theta}:=T(\bQ_{n+N},\pi)$. Similarly, the target $\target{\theta}$ \eqref{eq:targetfull} can be viewed as $T(\jointcov,\pi)$, where $\jointcov$ was defined as the joint distribution of $\rbr{X_i,C_i,Y_i}_{i=1}^{n+N}$ in the beginning of Section \ref{sec:covariate-shift}. Suppose the functional $T:\sQ \times \Pi \rightarrow \bR$ is jointly Hadamard differentiable tangentially to $L_2^0(\jointcov)\times L_2(\mathbb{Q}_X^*)$, where $\bQ_{n+N}, \jointcov\in \sQ$ and $\pi,\hat{\pi}\in \Pi$. 

We now provide an assumption on the quality of the estimator  $\hat{\pi}$. 
\begin{assumption} \label{assmp:estpi}
    Let $\hat{\pi}\in \Pi$ denote the estimator of $\pi$. We assume \[
    \hat{\pi}(x)-\pi(x)=\frac{1}{n+N}\sum_{i=1}^{n+N}\phi_\pi(C_i,X_i;x)+\littleOp\left(\frac{1}{\sqrt{n+N}}\right).
    \]
\end{assumption}
Assumption \ref{assmp:estpi} requires $\hat{\pi}$ to  be a regular ALE of $\pi$, with influence function $\phi_\pi$. This can be achieved if one is willing to assume some adequate structure on $\pi$. For example, if $\pi$ follows a parametric model and $\hat{\pi}$ is the maximum likelihood estimator from a correctly-specified parametric model, then Assumption \ref{assmp:estpi} is satisfied. 

We introduce the following notation:
\begin{align*}
    \fulllabmdestpi{\theta} := T(\bQ_{n+N},\hat{\pi}), \text{ and }
    \fulllabmdestpif{\theta}:= T(\bQ^f_{n+N},\hat{\pi}), \text{ where $\bQ^f_{n+N}$ is the empirical law of $(C_i,X_i,f(X_i))$. }
\end{align*}

\begin{remark}
    \label{prop:covshiftestpifull}
    $\fulllabmdestpi{\theta}$ is a regular ALE of $\target{\theta}$ (defined in \eqref{eq:targetfull}) with influence function $\varphi_\lab(c,x,y):=\frac{c}{\pi(x)}\phi(x,y)-\bE\sbr{\frac{C\phi(X,Y)}{\pi(X)^2}\phi_\pi(c,x;X)}$, and $\fulllabmdestpif{\theta}$ is also a regular ALE with influence function $\varphi_\lab(c,x,f(x)):=\frac{c}{\pi(x)}\phi(x,f(x))-\bE\sbr{\frac{C\phi(X,f(X))}{\pi(X)^2}\phi_\pi(c,x;X)}$, where $\IFtheta$ is defined in \eqref{eq:generalif_fdm} and $\phi_\pi$ is defined in Assumption \ref{assmp:estpi}. These results are proven in Appendix~\ref{app:proofrem4.1}. 
    Comparing these influence functions to those for $\fulllabmdest{\theta}$ and $\fulllabmdestf{\theta}$ in Remark \ref{rem:full_fdm}, estimating $\hat{\pi}$ contributes a first-order term to the influence functions, which can potentially lead to a reduction in variance relative to $\fulllabmdest{\theta}$. 
    \end{remark}

We propose a PPI estimator for $\target{\theta}$ as 
\begin{equation}\label{eq:ppiestpi}
    \hat{\theta}_{\hat{\pi}}:=\fulllabmdestpi{\theta}+\hat{\omega}\rbr{\hat{\theta}^f-\fulllabmdestpif{\theta}}.
\end{equation}
Our next result establishes  that its   asymptotic variance will be no worse than that of $\fulllabmdestpi{\theta}$, the estimator that does not utilize $\{f(X)\}_{i=1}^{n+N}$.

\begin{theorem}\label{thm:covshiftallestpi}(PPI for targets of the full data distribution with $\hat{\pi}$) Under Assumptions \ref{assmp:covdistshift}, \ref{assmp:estpi}, and the additional assumption that   $\hat{\omega}\xrightarrow{p}\omega$, where
     \begin{equation}
         \omega:=\frac{\Cov\Big(\varphi_\lab(C,X,Y),\varphi_\lab(C,X,f(X))-\IFtheta(X,f(X))\Big)}{\Var\Big(\varphi_\lab(C,X,f(X))-\IFtheta(X,f(X))\Big)},
     \end{equation} it follows that
\[
\sqrt{n+N}\rbr{\hat{\theta}_{\hat{\pi}}-\target{\theta}}\xrightarrow{d}N(0,\sigma^2),
\]
and $\sigma^2=\Var(\varphi_\lab(C,X,Y))-\frac{\Cov\Big(\varphi_\lab(C,X,Y),\;\IFtheta(X,f(X))-\varphi_\lab(C,X,f(X))\Big)^2}{\Var\Big(\IFtheta(X,f(X))-\varphi_\lab(C,X,f(X))\Big)}$, where $\varphi_\lab(C,X,Y)$ and $\varphi_\lab(C,X,f(X))$ are defined in Remark \ref{prop:covshiftestpifull}.   
\end{theorem}

The ideas in this section connect to the  AIPW estimators with estimated weights of \cite{robins1994estimation}. As in previous sections, even if $\IFtheta(x,y)$ is the efficient influence function, the estimator in \eqref{eq:ppiestpi} will still not be efficient, unless the (unrealistic) Assumption~\ref{assmp:oracle_case} is satisfied. However, here, \emph{we require neither assumptions on the quality of $f(X)$  nor an estimate of $\bE[\IFtheta(X,Y)\mid X]$}. Avoiding these requirements results in a loss of efficiency relative to the AIPW estimator that uses a ``good enough'' $\hat{\bE}[\IFtheta(X,Y)\mid X]$ in the rectifier (as was seen also in Sections \ref{sec:method} and \ref{sec:covariate-shift}). Additional details are provided in Appendix \ref{sec:appcomaipw}. However, in the spirit of this paper, the loss of efficiency  may be worthwhile, as it enables the straightforward utilization of the black-box machine learning model $f(X)$ in real-world settings where its quality cannot be guaranteed.

%% file: sec5binarymeric.tex
\section{Application to binary classifier metrics}
\label{sec:binaryclassifier}

In this section, we use the results from Sections~\ref{sec:method} and \ref{sec:covariate-shift} to develop PPI estimators for several common binary classification metrics: true positive rate (TPR), false positive rate (FPR), and area under the receiver operating characteristic curve (AUC). Suppose that $Y \in \{0,1\}$. The goal is to evaluate a biomarker $ R \in \sR$.
Given a threshold $\alpha \in [0,1]$,  the TPR, FPR, and AUC of $R$  are defined as 
\begin{align}
    \target{\theta}^\TPR &:= \Pr(R>\alpha\mid Y=1)=\frac{\bE[I(R>\alpha)Y]}{\bE[Y]}, \label{eq:tpr} \\
    \target{\theta}^\FPR &:= \Pr(R>\alpha\mid Y=0)=\frac{\bE[I(R>\alpha)(1-Y)]}{\bE[(1-Y)]}, \text{\;\;\;and}  \label{eq:fpr} \\
    \target{\theta}^\AUC &:= \int \Pr(R>\alpha\mid Y=1)d\Pr(R\leq\alpha\mid Y=0). \label{eq:auc}
\end{align}
As in previous sections, we have access to a deterministic prediction model $f$ such that $f(X) \approx Y$. Here, we assume $f(X)\geq 0$ and $\bE[f(X)]\big(1-\bE[f(X)]\big)\neq 0$ to satisfy the regularity conditions for the PPI estimator construction, as described in Section \ref{sec:method}. These assumptions are mild for a binary classification model. Without loss of generality, we assume that the vector $X$  contains $R$ as one of its elements. 

TPR, FPR, and AUC are not typically viewed through the lens of M-estimation, so the methods in \cite{angelopoulos2023ppi++, miao2025assumption,zrnic2024cross}, and \citet{gan2024prediction} do not directly apply. Instead, we will apply the ideas presented in earlier sections in order to develop PPI analogs of these quantities; to do this, we will make use of their influence functions, which can be derived via the functional delta method \cite[see, e.g.,][]{ledell2015computationally}.
%


\subsection{PPI for binary classification evaluation without covariate distribution shift}\label{sec:binary_no_covshift}

Suppose that $\rbr{R_i,f(X_i),Y_i}_{i=1}^n\iidsim\jointapp$, $\rbr{R_i,f(X_i)}_{i=n+1}^{n+N}\iidsim\marginalapp$, and $\jointapp=\conditionalapp\marginalapp$.
The following corollaries of Proposition~\ref{prop:PPest} provide PPI estimators for TPR, FPR, and AUC, respectively and describe their asymptotic distributions. 

\begin{corollary}[A PPI estimator for $\TPR$]\label{cor:tpr_ppi}
Define $\labest{\theta}:=\frac{\sum_{i=1}^nI(R_i>\alpha)Y_i}{\sum_{i=1}^nY_i}$ 
     and $\labfest{\theta}:=\frac{\sum_{i=1}^nI(R_i>\alpha)f(X_i)}{\sum_{i=1}^nf(X_i)}$. Let $\IFtheta^\TPR(R_i,Y_i)$ and  $\IFtheta^\TPR(R_i,f(X_i))$ denote their influence functions (evaluated at a single observation), respectively. Furthermore, let  $\allfest{\theta}:=\frac{\sum_{i=1}^{n+N} I(R_i>\alpha)f(X_i)}{\sum_{i=1}^{n+N} f(X_i)}.$ 


Then,  $\IFtheta^\TPR(R_i,Y_i)=\frac{I(R_i>\alpha)Y_i}{\bE\sbr{Y}}-\frac{\bE \sbr{I\rbr{R>\alpha}Y}Y_i}{\bE\sbr{Y}^2}$ and  $\IFtheta^\TPR(R_i,f(X_i))=\frac{I(R_i>\alpha)f(X_i)}{\bE\sbr{f(X)}}-\frac{\bE \sbr{I\rbr{R>\alpha}f(X)}f(X_i)}{\bE\sbr{f(X)}^2}$. Furthermore, defining 
\begin{equation*}
\ppest{\theta}:=\labest{\theta}+\hat{\omega}\rbr{\allfest{\theta}-\labfest{\theta}},
\label{eq:app-PPI}
\end{equation*} and supposing  
 that $\frac{n}{N}\rightarrow\lambda$ and 
$\hat{\omega}\xrightarrow{p}\frac{\bE[\IFtheta^\TPR(R,Y)\IFtheta^\TPR(R,f(X))]}{\bE[\IFtheta^\TPR(R,f(X))^2]},$
it follows that 
\begin{equation*}
    \sqrt{n}\rbr{\ppest{\theta}-\target{\theta}^{\TPR}}\indist \Norm\rbr{0,\bE[\IFtheta^{\TPR}(R,Y)^2]-\frac{\rbr{\bE[\IFtheta^{\TPR}(R,Y)\IFtheta^{\TPR}(R,f(X))]}^2}{(1+\lambda)\bE[\IFtheta^{\TPR}(R,f(X))^2]}}.
    \label{eq:PPI-asymptotic-tpr}
\end{equation*}
\end{corollary}

\begin{corollary}[A PPI estimator for $\FPR$]\label{cor:fpr_ppi}
Define $\labest{\theta}:=\frac{\sum_{i=1}^nI(R_i>\alpha)(1-Y_i)}{\sum_{i=1}^n(1-Y_i)}$ 
     and $\labfest{\theta}:=\frac{\sum_{i=1}^nI(R_i>\alpha)(1-f(X_i))}{\sum_{i=1}^n(1-f(X_i))}$. Let $\IFtheta^\FPR(R_i,Y_i)$ and  $\IFtheta^\FPR(R_i,f(X_i))$ denote their influence functions (evaluated at a single observation), respectively. Furthermore, define  $\allfest{\theta}:=\frac{\sum_{i=1}^{n+N} I(R_i>\alpha)(1-f(X_i))}{\sum_{i=1}^{n+N} (1-f(X_i))}.$ 
 
Then, $\IFtheta^\FPR(R_i,Y_i)=\frac{I(R_i>\alpha)(1-Y_i)}{\bE\sbr{1-Y}}-\frac{\bE \sbr{I\rbr{R>\alpha}(1-Y)}(1-Y_i)}{\bE\sbr{1-Y}^2}$ and  $\IFtheta^\FPR(R_i,f(X_i))=\frac{I(R_i>\alpha)(1-f(X_i))}{\bE\sbr{1-f(X)}}-\frac{\bE \sbr{I\rbr{R>\alpha}(1-f(X))}(1-f(X_i))}{\bE\sbr{1-f(X)}^2}$. Furthermore, defining 
\begin{equation*}
\ppest{\theta}:=\labest{\theta}+\hat{\omega}\rbr{\allfest{\theta}-\labfest{\theta}},
\label{eq:app-PPI}
\end{equation*} and supposing  
 that $\frac{n}{N}\rightarrow\lambda$ and 
$\hat{\omega}\xrightarrow{p}\frac{\bE[\IFtheta^\FPR(R,Y)\IFtheta^\FPR(R,f(X))]}{\bE[\IFtheta^\FPR(R,f(X))^2]},$
it follows that 
\begin{equation*}
    \sqrt{n}\rbr{\ppest{\theta}-\target{\theta}^{\FPR}}\indist \Norm\rbr{0,\bE[\IFtheta^{\FPR}(R,Y)^2]-\frac{\rbr{\bE[\IFtheta^{\FPR}(R,Y)\IFtheta^{\FPR}(R,f(X))]}^2}{(1+\lambda)\bE[\IFtheta^{\FPR}(R,f(X))^2]}}.
    \label{eq:PPI-asymptotic-tpr}
\end{equation*}
\end{corollary}

\begin{corollary}[A PPI estimator for AUC]\label{cor:auc_ppi}
Define $\labest{\theta}:=\int \frac{\sum_{i=1}^nI(R_i>\alpha)Y_i}{\sum_{i=1}^nY_i}d\rbr{\frac{\sum_{i=1}^nI(R_i\leq\alpha)(1-Y_i)}{\sum_{i=1}^n(1-Y_i)}}$ 
     and $\labfest{\theta}:=\int \frac{\sum_{i=1}^nI(R_i>\alpha)f(X_i)}{\sum_{i=1}^nf(X_i)}d\rbr{\frac{\sum_{i=1}^nI(R_i\leq\alpha)(1-f(X_i))}{\sum_{i=1}^n(1-f(X_i))}}$. Let $\IFtheta^\AUC(R_i,Y_i)$ and  $\IFtheta^\AUC(R_i,f(X_i))$ denote their influence functions (evaluated at a single observation), respectively. Furthermore, define 
     $\allfest{\theta}:=\int \frac{\sum_{i=1}^{n+N}I(R_i>\alpha)f(X_i)}{\sum_{i=1}^{n+N}f(X_i)}d\rbr{\frac{\sum_{i=1}^{n+N}I(R_i\leq\alpha)(1-f(X_i))}{\sum_{i=1}^{n+N}(1-f(X_i))}}.$
 
Then, \[\IFtheta^\AUC(R_i,Y_i)=\frac{1-Y_i}{\bE[1-Y]}\left(1-\frac{\bE[I(R\leq\alpha)Y]}{\bE[Y]}\right)\bigg|_{\alpha=R_i}+\frac{Y_i}{\bE[Y]}\frac{\bE[I(R\leq\alpha)(1-Y)]}{\bE[1-Y]}\bigg|_{\alpha=R_i}\nonumber\]
\[-\left(\frac{1-Y_i}{\bE[1-Y]}+\frac{Y_i}{\bE[Y]}\right)
   \int \frac{\bE[I(R>\alpha)Y]}{\bE[Y]}d\rbr{\frac{\bE[I(R\leq \alpha)(1-Y)]}{\bE[1-Y]}},\] 
and  \[\IFtheta^\AUC(R_i,f(X_i))=\frac{1-f(X_i)}{\bE[1-f(X)]}\left(1-\frac{\bE[I(R\leq\alpha)f(X)]}{\bE[f(X)]}\right)\bigg|_{\alpha=R_i}+\frac{f(X_i)}{\bE[f(X)]}\frac{\bE[I(R\leq\alpha)(1-f(X))]}{\bE[1-f(X)]}\bigg|_{\alpha=R_i}\nonumber\]\[-\left(\frac{1-f(X_i)}{\bE[1-f(X)]}+\frac{f(X_i)}{\bE[f(X)]}\right)
   \int \frac{\bE[I(R>\alpha)f(X)]}{\bE[f(X)]}d\rbr{\frac{\bE[I(R\leq \alpha)(1-f(X))]}{\bE[1-f(X)]}}.\] 
   Furthermore, defining 
\begin{equation*}
\ppest{\theta}:=\labest{\theta}^\AUC+\hat{\omega}\rbr{\allfest{\theta}-\labfest{\theta}},
\label{eq:app-PPI}
\end{equation*} and supposing  
 that $\frac{n}{N}\rightarrow\lambda$ and 
$\hat{\omega}\xrightarrow{p}\frac{\bE[\IFtheta^\AUC(R,Y)\IFtheta^\AUC(R,f(X))]}{\bE[\IFtheta^\AUC(R,f(X))^2]},$
it follows that 
\begin{equation*}
    \sqrt{n}\rbr{\ppest{\theta}-\target{\theta}^{\AUC}}\indist \Norm\rbr{0,\bE[\IFtheta^{\AUC}(R,Y)^2]-\frac{\rbr{\bE[\IFtheta^{\AUC}(R,Y)\IFtheta^{\AUC}(R,f(X))]}^2}{(1+\lambda)\bE[\IFtheta^{\AUC}(R,f(X))^2]}}.
    \label{eq:PPI-asymptotic-tpr}
\end{equation*}
\end{corollary}

\subsection{PPI for binary classification metrics with covariate distribution shift}\label{sec:appcovshift}
Next, we consider the covariate distribution shift setting of Section \ref{sec:covdistall}, in the contexts of TPR and AUC. Thus, the targets are defined as in \eqref{eq:tpr} and \eqref{eq:auc}, where the probabilities are taken with respect to the full data distribution, as described in Section \ref{sec:covariate-shift}.  Results for FPR \eqref{eq:fpr} or for the scenarios in Sections \ref{sec:covdistunlab} and \ref{sec:covdistlab} can be derived analogously. 

Recall that we have assumed, without loss of generality,  that $R$ is an element of $X$.
 Define $\pi(X)=\Pr(C=1\mid X)$; in this illustration, we assume $\pi$ is known and satisfies Assumption \ref{assmp:covdistshift}. The next two results are corollaries of Theorem \ref{thm:ppiipw}. When $\pi$ is unknown, similar results follow from Theorem \ref{thm:covshiftallestpi}.

\begin{corollary}[A PPI estimator for $\TPR$ of the full distribution]\label{cor:tpr_ppi_cov}
    Define
    $\fulllabmdest{\theta}:=\frac{\sum_{i=1}^{n+N}\frac{C_i}{\pi(X_i)}I(R_i\leq\alpha)Y_i}{\sum_{i=1}^{n+N}\frac{C_i}{\pi(X_i)}Y_i}$, $\fulllabmdestf{\theta}:=\frac{\sum_{i=1}^{n+N}\frac{C_i}{\pi(X_i)}I(R_i\leq\alpha)f(X_i)}{\sum_{i=1}^{n+N}\frac{C_i}{\pi(X_i)}f(X_i)}$, and 
    $\fullmdestf{\theta}:=\frac{\sum_{i=1}^{n+N}I(R_i\leq\alpha)f(X_i)}{\sum_{i=1}^{n+N}f(X_i)}$.
    
Suppose
$\hat{\omega}\inprob\frac{\Cov\left(\frac{C_i}{\pi(X_i)}\IFtheta^\TPR(R_i,Y_i),\left(\frac{C_i}{\pi(X_i)}-1\right)\IFtheta^\TPR(R_i,f(X_i))\right)}{\Var\left(\rbr{1-\frac{C_i}{\pi(X_i)}}\IFtheta^\TPR(R_i,f(X_i))\right)}$, where $\IFtheta^\TPR(R_i,Y_i)$ and $\IFtheta^\TPR(R_i,f(X_i))$ are given in Corollary \ref{cor:tpr_ppi}.  
Define\[
    \fullest{\theta}:=\fulllabmdest{\theta}+\hat{\omega}\rbr{\fullmdestf{\theta}-\fulllabmdestf{\theta}}.
\]
Under Assumption \ref{assmp:covdistshift},
we have
$$\sqrt{n+N}\rbr{\fullest{\theta}-\target{\theta}^\TPR}\indist N(0,\sigma^2),$$ where $\sigma^2=\Var\rbr{\frac{C_i}{\pi(X_i)}\IFtheta^\TPR(R_i,Y_i)}-\frac{\Cov\rbr{\frac{C_i}{\pi(X_i)}\IFtheta^\TPR(R_i,Y_i),\rbr{1-\frac{C_i}{\pi(X_i)}}\IFtheta^\TPR(R_i,f(X_i))}^2}{\Var\rbr{\rbr{1-\frac{C_i}{\pi(X_i)}}\IFtheta^\TPR(R_i,f(X_i))}}$.
\end{corollary}

\begin{corollary}[A PPI estimator for AUCof the full distribution]\label{cor:auc_ppi_cov}
    Define
    $\fulllabmdest{\theta}:=\int (1-F_{\labfull,1}(\alpha))dF_{\labfull,0}(\alpha)$, $\fulllabmdestf{\theta}:=\int (1-F^{f}_{\labfull,1}(\alpha))dF^{f}_{\labfull,0}(\alpha)$, and 
    $\fullmdestf{\theta}:=\int (1-F^{f}_{1}(\alpha))dF^{f}_{0}(\alpha)$, where
\begin{align*}
    F_{\labfull,1}(\alpha)&:= \frac{\sum_{i=1}^{n+N}\frac{C_i}{\pi(X_i)}I(R_i\leq\alpha)Y_i}{\sum_{i=1}^{n+N}\frac{C_i}{\pi(X_i)}Y_i},
    &F_{\labfull,0}(\alpha):=\frac{\sum_{i=1}^{n+N}\frac{C_i}{\pi(X_i)}I(R_i\leq\alpha)(1-Y_i)}{\sum_{i=1}^{n+N}\frac{C_i}{\pi(X_i)}(1-Y_i)},\\
    F^{f}_{\labfull,1}(\alpha)&:= \frac{\sum_{i=1}^{n+N}\frac{C_i}{\pi(X_i)}I(R_i\leq\alpha)f(X_i)}{\sum_{i=1}^{n+N}\frac{C_i}{\pi(X_i)}f(X_i)},
    &F^{f}_{\labfull,0}(\alpha):=\frac{\sum_{i=1}^{n+N}\frac{C_i}{\pi(X_i)}I(R_i\leq\alpha)(1-f(X_i))}{\sum_{i=1}^{n+N}\frac{C_i}{\pi(X_i)}(1-f(X_i))},\\
    F^{f}_{1}(\alpha)&:= \frac{\sum_{i=1}^{n+N}I(R_i\leq\alpha)f(X_i)}{\sum_{i=1}^{n+N}f(X_i)},
    &F^{f}_{0}(\alpha):=\frac{\sum_{i=1}^{n+N}I(R_i\leq\alpha)(1-f(X_i))}{\sum_{i=1}^{n+N}(1-f(X_i))}.
\end{align*}
Suppose
$\hat{\omega}\inprob\frac{\Cov\left(\frac{C_i}{\pi(X_i)}\IFtheta^\AUC(R_i,Y_i),\left(\frac{C_i}{\pi(X_i)}-1\right)\IFtheta^\AUC(R_i,f(X_i))\right)}{\Var\left(\rbr{1-\frac{C_i}{\pi(X_i)}}\IFtheta^\AUC(R_i,f(X_i))\right)}$, where $\IFtheta^\AUC(R_i,Y_i)$ and  $\IFtheta^\AUC(R_i,f(X_i))$ are given in Corollary \ref{cor:auc_ppi}.
Define \[
    \fullest{\theta}:=\fulllabmdest{\theta}+\hat{\omega}\rbr{\fullmdestf{\theta}-\fulllabmdestf{\theta}}.
\]
Under Assumption \ref{assmp:covdistshift},
we have
$$\sqrt{n+N}\rbr{\fullest{\theta}-\target{\theta}^\AUC}\indist N(0,\sigma^2),$$ where $\sigma^2=\Var\rbr{\frac{C_i}{\pi(X_i)}\IFtheta^\AUC(R_i,Y_i)}-\frac{\Cov\rbr{\frac{C_i}{\pi(X_i)}\IFtheta^\AUC(R_i,Y_i),\rbr{1-\frac{C_i}{\pi(X_i)}}\IFtheta^\AUC(R_i,f(X_i))}^2}{\Var\rbr{\rbr{1-\frac{C_i}{\pi(X_i)}}\IFtheta^\AUC(R_i,f(X_i))}}$.
\end{corollary}

\subsection{The special case of evaluating a prediction model}

\label{rem:fequalsR}

In Sections \ref{sec:binary_no_covshift} and \ref{sec:appcovshift}, we assumed that $f(X)$ was an arbitrary machine learning model and that our goal was to evaluate the performance of a biomarker $R$. Here, instead, we suppose that our goal is to estimate the performance of $f(X)$ as a prediction model for $Y$. That is, in this subsection we take $f(X)=R$. 


 In this setting, the target (TPR, FPR, or AUC) itself involves $f(X)$. Under modest regularity conditions (see the beginning of Section~\ref{sec:rectifier}), and provided that the influence functions $\IFtheta^{\TPR}(f(X),f(X))$ and $\IFtheta^{\TPR}(f(X),Y)$ are not orthogonal (analogously for FPR and AUC), the PPI estimator will provide a reduction in the asymptotic variance relative to the estimator that uses only labeled data. This reduction will tend to be larger when $f(X)$ is a good approximation of $Y$. This indicates that the PPI estimators of TPR, FPR, and AUC do not actually require an additional black-box prediction model to achieve an improvement in the context of evaluating a given prediction model. 
    

    However, the same is not true for all estimands. To illustrate this, we consider the mean squared error (MSE), 
    defined as $\target{\theta}^\text{MSE}(Y,R):=E\left[ (Y-R)^2 \right]$. 
    A PPI estimator of MSE motivated by Proposition~\ref{prop:PPest} will improve upon the asymptotic variance of the labeled estimator,  provided that the influence functions $\IFtheta^{\text{MSE}}(R,Y)$ and $\IFtheta^{\text{MSE}}(R,f(X))$ are not orthogonal,  where $\IFtheta^{\text{MSE}}(R,f(X))=(f(X)-R)^2-E\left[ (f(X)-R)^2 \right]$. However, in the setting of $f(X)=R$, because $\IFtheta^{\text{MSE}}(f(X),f(X))=0$, a PPI estimator using $f(X)$ in its rectifier term will never improve the estimation of the MSE of $f(X)$, compared to the labeled data estimator, no matter how well $f(X)$ approximates $Y$. To see this directly, note that a PPI estimator for MSE takes the form
$$\frac{1}{n}\sum_{i=1}^n (Y_i - R_i)^2 + \hat\omega \left( \frac{1}{n+N}\sum_{i=1}^{n+N}(f(X_i)-R_i)^2-\frac{1}{n}\sum_{i=1}^n (f(X_i)-R_i)^2 \right), $$ for some suitably-defined $\hat\omega$. In the case of $f(X)=R$,  the  rectifier is exactly zero.

%% file: sec6simulations.tex
\section{Simulations}\label{sec:simulation}
\subsection{Simulations for PPI without covariate distribution shift}\label{sec:sim_nocovshift}

We now numerically investigate the proposed PPI estimators for TPR, FPR, and AUC derived in Section \ref{sec:binaryclassifier}. As a point of comparison, we also investigate the PPI estimator for the mean, $\bE(Y)$. Additional simulation results under a different data generation process, and illustrating Section \ref{rem:fequalsR}, are in Appendix \ref{sec:xtrasims}.


\subsubsection{Data generation} \label{sec:datagen}
We generate data with covariates $\textbf{X}:=\rbr{R,X_1,X_2,X_3,X_4,X_5}^\top\in \bR^6$ and outcome $Y \in \cbr{0,1}$ following\\
$\begin{pmatrix}
R  \\
X_1\\
X_2\\
X_3\\
X_4\\
X_5
\end{pmatrix}\iidsim\text{MVN}\rbr{\textbf{0},\Sigma}$, where $\Sigma=\begin{bmatrix}
1 & 0.5 & 0.3 & 0.2 & 0 & 0 \\
0.5 & 1 & 0.4 & 0.3 & 0 & 0 \\
0.3 & 0.4 & 1 & 0.2 & 0 & 0 \\
0.2 & 0.3 & 0.2 & 1 & 0 & 0 \\
0 & 0 & 0 & 0 & 1 & 0 \\
0 & 0 & 0 & 0 & 0 & 2
\end{bmatrix}$, 
$Y\mid \textbf{X}\overset{\text{ind}}{\sim} \text{Bernoulli}(\Pr(Y=1\mid \textbf{X}))$, and
logit$(\Pr(Y=1\mid \textbf{X}))=-0.5+R-0.9X_1^2+0.6|X_2|+0.5X_3^3+1.5R X_3-0.7X_1 X_2$.

We investigate five candidate prediction models for $f$:
\begin{itemize}
    \item $f(\textbf{X}) = \Pr(Y=1\mid \textbf{X})$: This is the ideal model, i.e., $f(\textbf{X})=\bE[Y\mid \textbf{X}]$. Based on the form of the influence functions shown in Section \ref{sec:binary_no_covshift}, Assumption \ref{assmp:oracle_case} is satisfied for this choice of $f$ when estimating TPR, FPR, and AUC, as well as for the mean.
    \item $f(\textbf{X}) = \text{RF}(Y\sim \textbf{X})$: A random forest model predicting $Y$ from $\textbf{X}$.
    \item $f(\textbf{X}) = \text{RF}(Y\sim R+X_1+X_4+X_5)$: A random forest model predicting $Y$ from $R, X_1,X_4$, and $X_5$.
    \item $f(\textbf{X}) = \text{RF}(Y\sim X_4+X_5)$: A random forest model predicting $Y$ only from $X_4$ and $X_5$, which are independent of $Y$.
    \item $f(\textbf{X}) = \text{Unif}[0.01,0.99]$: Uniformly-distributed random noise. 
\end{itemize}
These prediction models are independent of the data for conducting inference. Specifically, the random forest models are trained on 100,000 independent observations.

We set $n=1000$ and vary $\lambda$, where $\lambda=\frac{n}{N}\in \cbr{0.01,0.1,0.25,0.5,0.8}$. We  generate the outcomes $Y$ for $n$ observations, but we generate $\textbf{X}$ for all $n+N$ observations. The simulation results are evaluated over 2500 replications, providing a reasonable bound on the Monte Carlo error of estimates of confidence interval coverage \citep{morris2019using}. Unless otherwise specified, we report mean values over the 2500 replications.

\subsubsection{Simulation results}\label{sec:simres_noncovshift}

\begin{figure}
    \centering
    \includegraphics[width=\linewidth]{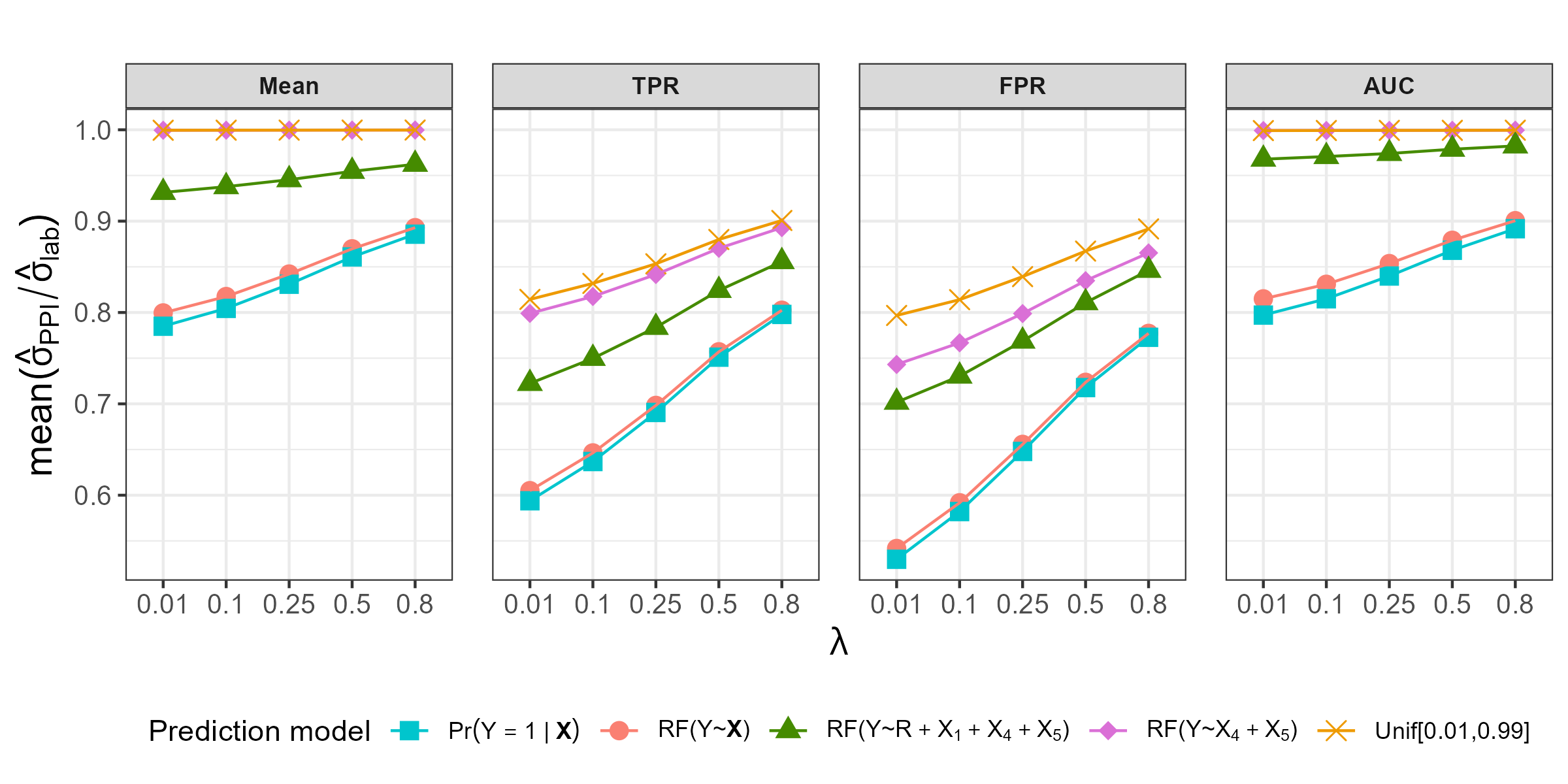}
    \caption{Ratio of estimated standard error for the PPI estimator to that of the estimator using only labeled data, averaged over 2500 replications, with $n=1000,$ $\lambda\in\cbr{0.01,0.1,0.25,0.5,0.8}$, for mean, TPR, FPR, and AUC estimation. Shapes and colors denote the different prediction models $f$ used in PPI. The standard errors of PPI estimators are no greater than those of the estimators using only the labeled data. The ideal model, $\Pr(Y=1\mid \textbf{X})$ (squares), provides the most improvement for every estimand. The noisy prediction models, $\text{RF}(Y\sim X_4+X_5)$ (diamonds) and $\text{Unif}[0.01,0.99]$ (crossed lines), have no improvement for mean or AUC estimation, but do result in a variance decrease in TPR and FPR estimation.}
    \label{fig:ratio_if}
\end{figure}
 Figure \ref{fig:ratio_if} compares the estimated standard error of our PPI estimators and that of the estimators using only the labeled data for the following four targets: the mean of $Y$, and the TPR, FPR, and AUC of $R$, as a biomarker or a predictor of $Y$. Both TPR and FPR are evaluated at a threshold value of 0.6. 
 
Though the ideal prediction model, $\Pr(Y=1\mid \textbf{X})$, exhibits the most improvement over the estimator that only uses the labeled observations, the model  $\text{RF}(Y\sim \textbf{X})$ performs almost as well. The model $\text{RF}(Y\sim R+X_1+X_4+X_5)$ performs a bit worse, since it does not contain two of the relevant predictors. 

Next, we turn to the noisy prediction models, $\text{RF}(Y\sim X_4+X_5)$ and $\text{Unif}[0.01,0.99]$. For these prediction models, $\Cov(f(\textbf{X}),Y)=0$. One might expect to see no improvement over the estimator that uses only the labeled observations. However, the PPI estimator does exhibit a variance decrease in the case of TPR and FPR, though not for mean and AUC. To understand this, recall that a PPI estimator provides a reduction in asymptotic variance if and only if $\IFtheta(X,Y)$ and $\IFtheta(X, f(X))$ are not orthogonal. In the case of the mean and AUC, $\Cov(f(\textbf{X}),Y)=0$ implies that $\Cov(\IFtheta^\text{Mean}(X, f(X)),\IFtheta^\text{Mean}(X, Y))$ and $\Cov(\IFtheta^\AUC(X, f(X)),\IFtheta^\AUC(X, Y))$ equal $0$, and so no variance reduction occurs. By contrast, for FPR and TPR, $\Cov(f(\textbf{X}),Y)=0$ does not imply that $\Cov(\IFtheta^\TPR(X, f(X)),\IFtheta^\TPR(X, Y))$ or $\Cov(\IFtheta^\FPR(X, f(X)),\IFtheta^\FPR(X, Y))$ equals $0$, and so there is a variance reduction (details are provided in Appendix \ref{sec:reducintprfpr}).
Furthermore,  the reduction in variance arising from  $\text{RF}(Y\sim X_4+X_5)$ is greater than that arising from $\text{Unif}[0.01,0.99]$. This is because the former captures some information about $Y$ --- for example, $\text{RF}(Y\sim X_4+X_5)$ approximates $\bE[Y]$ --- and thereby may lead to a greater $\Cov(\IFtheta^\TPR(X,Y),\IFtheta^\TPR(X,f(X)))^2$ and $\Cov(\IFtheta^\FPR(X,Y),\IFtheta^\FPR(X,f(X)))^2$ relative to the latter. 

Across all estimands, for all forms of $f(X)$ that exhibit improvement in variance for PPI estimators, smaller values of $\lambda$ lead to more improvement in the PPI estimators. This is because the variance decrease arising from PPI is proportional to $\frac{1}{1+\lambda}$, as demonstrated in Section \ref{sec:rectifier}.

 \begin{figure}[htbp]
    \centering
    \includegraphics[width=1\linewidth]{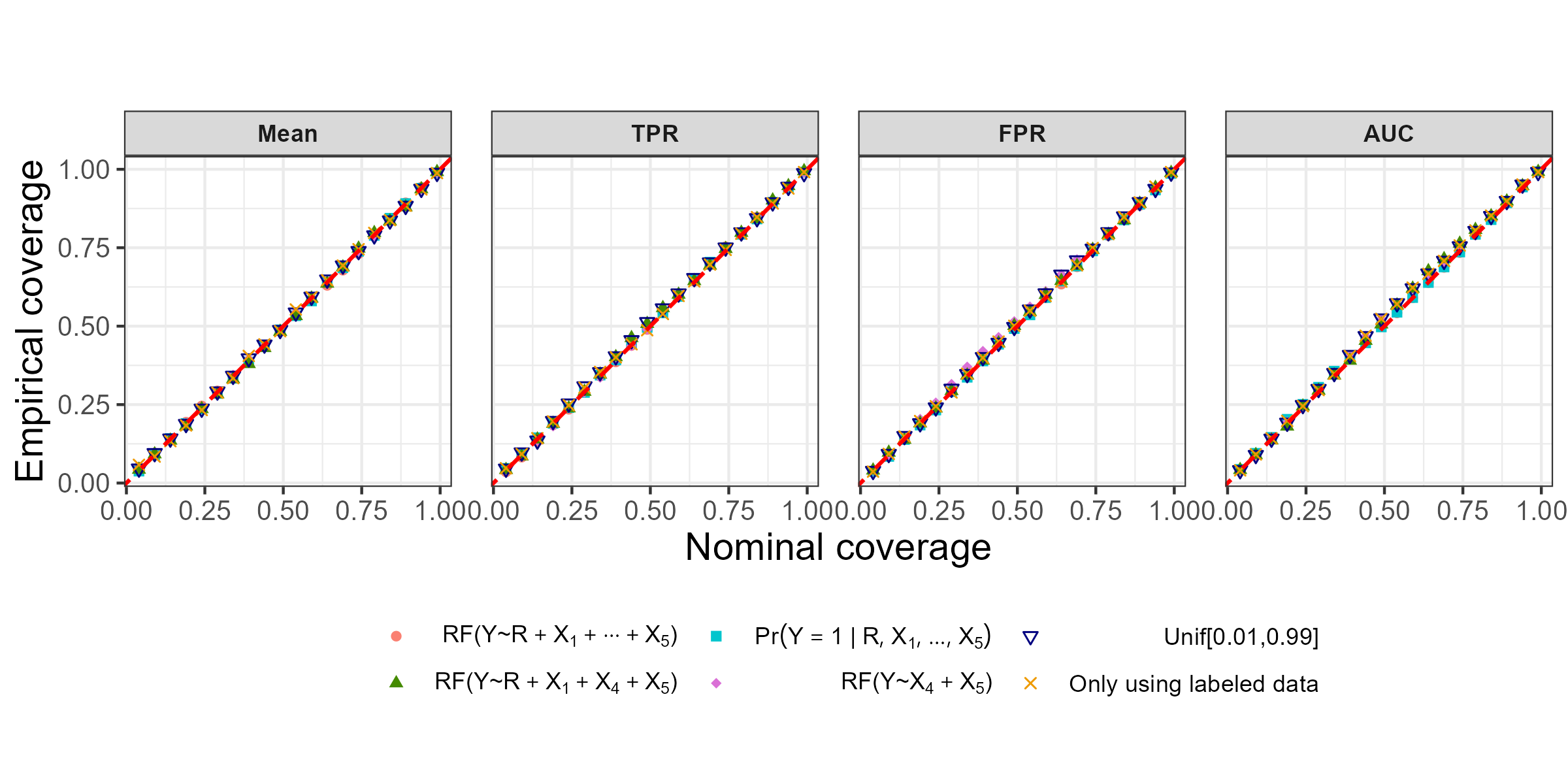}
    \caption{Empirical coverage versus nominal coverage, averaged over 2500 replications, with $n=1000$, $\lambda=0.1$, in the context of mean, TPR, FPR, and AUC estimation. Shapes denote different prediction models $f$ used in PPI. The red dashed line is the diagonal line. The empirical coverages are close to the nominal coverages. }
    \label{fig:coverage}
\end{figure}

Figure \ref{fig:coverage} compares empirical coverage to the corresponding nominal coverage for the five PPI estimators, along with the estimator using only the labeled data, in the context of mean, TPR, FPR, and AUC estimation with $n=1000$ and  $\lambda=0.1$. The empirical coverage is very close to the nominal coverage. We also ran simulations with $\lambda=0.1$ and $n=10000$ over $10000$ replications, and the coverage results again align well with the nominal coverage (see Figure \ref{fig:big_cov}).

\subsection{Simulations under covariate distribution shift}\label{sec:sim_covshift}
We now numerically investigate the proposed PPI estimators for TPR, FPR, and AUC of $R$, as a biomarker or a predictor of $Y$, under the setting of Section \ref{sec:covdistall}, where we consider a target of the full data distribution under covariate shift. We also present the results for estimation of $\bE[Y]$ for comparison.

\subsubsection{Data generation}
In this section, both $\textbf{X}:=\rbr{R,X_1,X_2,X_3,X_4,X_5}^\top\in \bR^6$ and  $Y \in \cbr{0,1}$ are generated exactly as in Section \ref{sec:datagen}.

We set  $\pi(\mathbf X)= 0.2 + 0.6\,\frac{\exp\!\big(R - 0.9X_{1} + 0.7X_{2}X_{3} - 0.5\big)}{1 + \exp\!\big(R - 0.9X_{1} + 0.7X_{2}X_{3} - 0.5\big)}$, where $\pi(\textbf{X}) = \Pr(C=1 \mid \textbf{X})$; 
thus, we generate the indicator $C\in \{0,1\}$ following $C\mid \textbf{X}\overset{\text{ind}}{\sim} \text{Bernoulli}(\pi(\textbf{X}))$. In this section, we assume that $\pi(\cdot)$ is known. However, $\pi(\cdot)$ can also be estimated, and provided that the estimator satisfies Assumption \ref{assmp:estpi}, results in Section \ref{sec:estpi} can be applied.

We investigate the same five candidate prediction models for $f$ as in Section \ref{sec:datagen}. We set the sample size $n+N$ for each simulated dataset to 50,000, and perform 2,500 replicates of the simulation.

\subsubsection{Simulation results}\label{sec:simres_covshift}

Figure \ref{fig:ratio_if_covshift} compares the estimated standard errors of our PPI estimators to that  of the estimator $\labmdest{\theta}$ defined in Section~\ref{sec:covdistall}, for the following four targets: the mean of $Y$, and the TPR, FPR, and AUC of $R$, as a biomarker or a predictor of $Y$. Both TPR and FPR are evaluated at a threshold value of $\alpha=0.6$.

\begin{figure}[htbp]
    \centering
    \includegraphics[width=\linewidth]{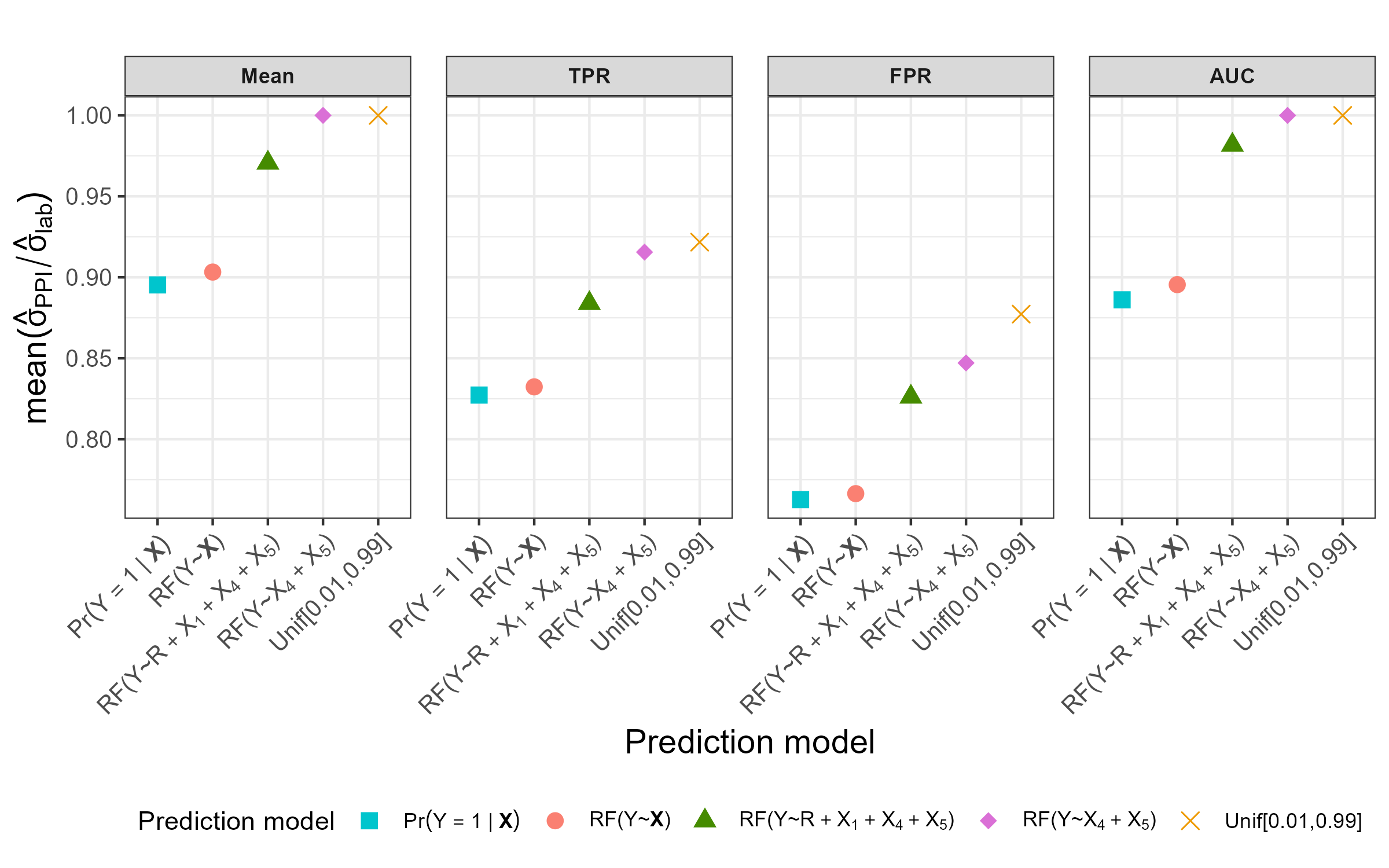}
    \caption{Ratio of estimated standard error for the PPI estimator to that of the estimator $\labmdest{\theta}$ defined in Section~\ref{sec:covdistall}, averaged over 2500 replications, with $n+N=50,000$, for mean, TPR, FPR, and AUC estimation, in the settings of Section \ref{sec:covdistall}.  Shapes denote different prediction models $f$ used in PPI. The standard errors of the PPI estimators are no greater than those of  $\labmdest{\theta}$. The ideal model, $\Pr(Y=1\mid \textbf{X})$ (squares), provides the most improvement for every estimand. The noisy prediction models, $\text{RF}(Y\sim X_4+X_5)$ (diamonds) and $\text{Unif}[0.01,0.99]$ (crossed lines), have no improvement for mean or AUC estimation, but do result in a variance decrease in TPR and FPR estimation.} 
    \label{fig:ratio_if_covshift}
\end{figure}

As expected, the PPI estimators that make use of the ideal model, $\Pr(Y=1\mid \textbf{X})$, exhibit the most improvement relative to $\labmdest{\theta}$, and the ones that use the model RF($Y\sim \textbf{X}$) perform almost as well. The PPI estimators that use the model RF($Y\sim R+X_1+X_4+X_5$) do not perform as well, since this model omits two of the relevant predictors. Similarly, the PPI estimators with the noisy prediction models, RF($Y\sim X_4+X_5$) and Unif[0.01, 0.99], exhibit no improvement over $\labmdest{\theta}$ for mean and AUC estimation, and improve on $\labmdest{\theta}$  for TPR and FPR estimation. This was discussed in Appendix \ref{sec:reducintprfpr} in the context of Figure \ref{fig:ratio_if}, under the assumption that $C\indep Y\mid X$; an analogous argument applies here.

\begin{figure}[htbp]
    \centering
    \includegraphics[width=\linewidth]{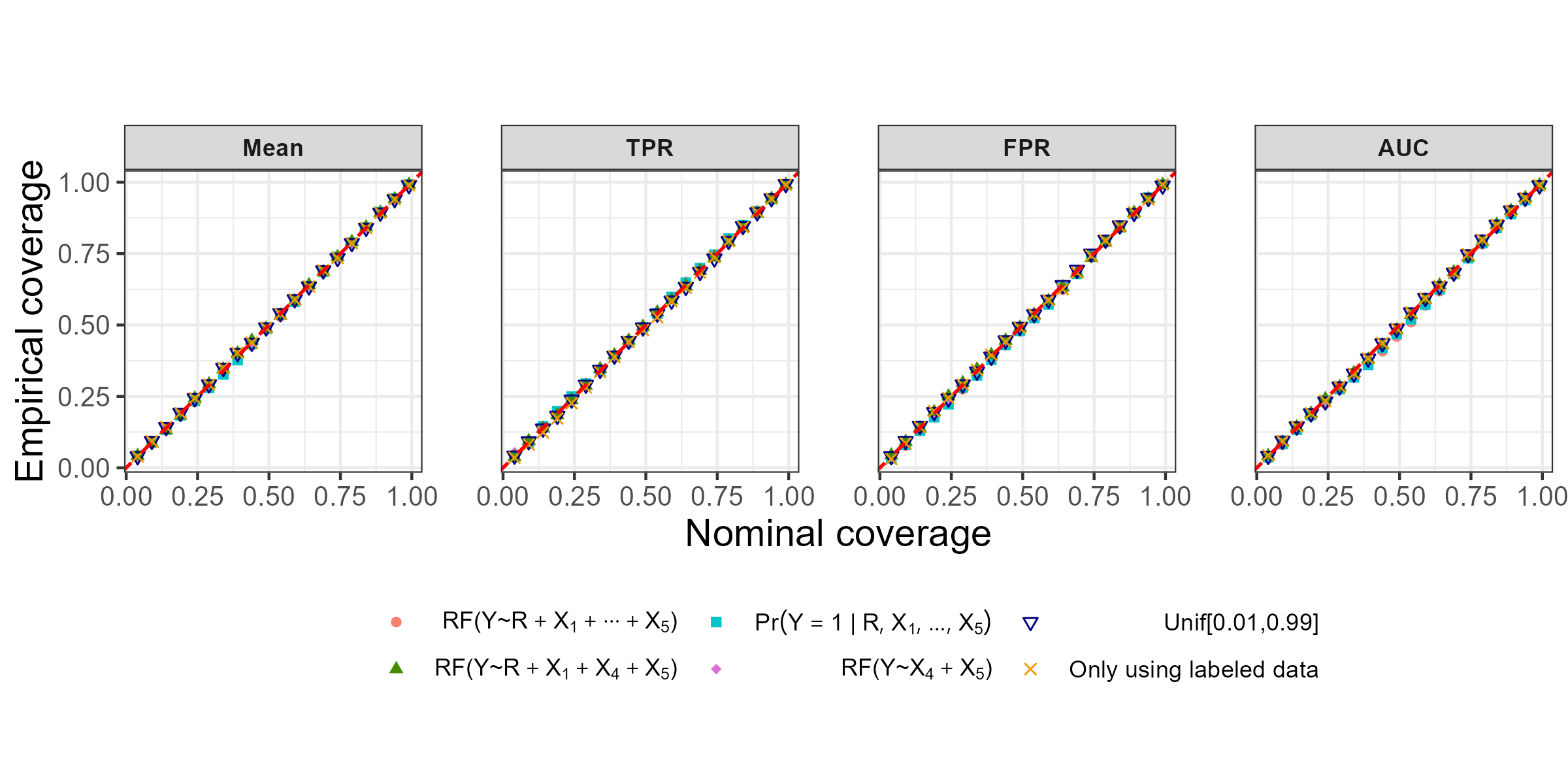}
    \caption{Empirical coverage versus nominal coverage, averaged over 2500 replications, with $n+N=50,000$, for mean, TPR, FPR, and AUC estimation, in the setting of Section \ref{sec:covdistall}. Shapes denote different prediction models $f$ used in PPI. The red dashed line is the diagonal line. The empirical coverages are close to the nominal coverage. }
    \label{fig:cov_covshift_mc}
\end{figure}
Figure \ref{fig:cov_covshift_mc}  displays the empirical coverage of the PPI estimators for mean, TPR, FPR, and AUC estimation. In each case, the empirical coverage is close to the nominal coverage. 

We present additional simulation results in Section~\ref{sec:xtrasims}. 
In Appendix \ref{sec:simforrem}, we provide numerical results showing that PPI can be used in the setting where $f(X) = R$ (as discussed in Section \ref{rem:fequalsR}). In Appendix \ref{sec:cordgp}, we run simulations under a different data-generating distribution. In Appendix \ref{sec:big_sim}, we show the coverage results at a much larger sample size. We also compute the bias of the PPI estimators in the settings of Sections \ref{sec:sim_nocovshift} and \ref{sec:sim_covshift} in Figures \ref{fig:bias_nocovshift} and \ref{fig:bias_covshift}, respectively, in Appendix \ref{sec:biasres}. In all cases, our procedure performs as expected from the theory. 

%% file: sec7realdata.tex
\section{Analyzing wine quality using PPI}\label{sec:realdata}
The two wine quality datasets from the UCI machine learning repository \citep{cortez2009modeling} relate to red and white variants of the Portuguese ``Vinho Verde" wine. We combine the two datasets to obtain a total of 1599 red wines  and 4898 white wines, along with 12 physicochemical features. More details about the data can be found in \cite{cortez2009modeling}. Our goal is to investigate the TPR, FPR, and AUC of  wine density (a continuous variable ranging from 0.987 to 1.039 g/cm$^3$) as a biomarker of wine type (red or white, coded as 1 or 0, respectively), and to determine whether the variance of these estimates can be reduced using PPI. 

We randomly sample 2000 observations from the combined dataset as an external model-building dataset to fit the prediction model $f(\cdot)$, which we obtain using logistic regression with all features. We leave all the remaining observations as our ``evaluation dataset''.  The AUC of  $f(\cdot)$ on the evaluation dataset is 0.996.
All further computations are performed on the evaluation dataset.

To determine which observations in the evaluation dataset are ``labeled", we use the ``quality" feature: if quality $\leq 6$ then the probability of labeling is 0.2, and otherwise it is 0.3.  As our interest lies in the TPR, FPR, and AUC of the full data distribution, we will make use of the PPI estimator from Section~\ref{sec:covdistall}.

Table \ref{tb:realdatares} displays the results of estimating TPR and FPR of density as a biomarker of wine type (both calculated with a threshold of $0.998$ g/cm$^3$), as well as AUC, for (i) the estimator that does not use the predictions in the unlabeled data, (ii) the PPI estimator, and (iii) the ideal estimator that we would be able to compute if we had labels for the unlabeled data (as opposed to just for the labeled data). Of course, (iii) is not possible in practice, but it enables us to see how much efficiency is lost by PPI relative to the ideal setting where all data were labeled. In this case, we see that the estimated standard error  of the PPI estimator (ii) is close to the estimated standard error of the  (unattainable) ideal estimator (iii), and is approximately half of the estimated standard error of the labeled-data estimator (i).
\begin{table}[ht] 
\centering
\begin{tabular}{l | ccc | ccc}
\hline
 & \multicolumn{3}{c|}{Point Estimates} & \multicolumn{3}{c}{Standard Error Estimates} \\
\cline{2-7}
 & (i) & (ii) & (iii) & (i) & (ii) & (iii) \\
\hline
TPR & 0.2141 & 0.2035 &  0.2018 & 0.02595 & 0.01244 &  0.01210\\
FPR & 0.0969 & 0.1096 &  0.1107 & 0.01110 & 0.00541 &  0.00538\\
AUC & 0.7903 & 0.7724 &  0.7750 & 0.01529 & 0.00817 & 0.00689 \\
\hline
\end{tabular}
\caption{Estimates of TPR, FPR, and AUC, along with estimated standard errors, for (i) the estimator that does not use the predictions in the unlabeled data, (ii) the PPI estimator, and (iii) the ideal estimator that we would be able to compute if we had labels for the unlabeled data.}\label{tb:realdatares}
\end{table}

%% file: sec8discussion.tex
\section{Discussion}\label{sec:discussion}  
In this paper, we proposed a generalized framework to construct and analyze PPI estimators derived from any  regular asymptotically linear estimator, a class that extends well beyond the relatively limited setting of M-estimation previously considered in the PPI literature \citep{angelopoulos2023prediction,angelopoulos2023ppi++}. 
We further generalized this framework to accommodate three  distinct scenarios for covariate distribution shift. Our framework yields PPI estimators that are guaranteed to have  asymptotic variance no greater than that of an estimator that ignores the unlabeled data.

We have also clarified the motivation for the form of the rectifier used in PPI estimators, and made explicit connections to the AIPW and missing data literature \citep{robins1994estimation,robins1995semiparametric,tsiatis2006semiparametric}. While the PPI estimators will in general not achieve the semi-parametric efficiency bound, there is a trade-off between optimal variance and simplicity: PPI does not require estimating the conditional expectation of the (efficient) influence function \citep{robins1994estimation,tsiatis2006semiparametric}, nor does it require either Donsker class conditions or cross-fitting for valid inference \citep{xu2025unified,testa2025semiparametric}. 


There remain several directions for future work. While the  PPI estimator will in general not achieve the semiparametric efficiency bound, the difference in finite-sample performance between the PPI estimator and estimators that can achieve this bound remains to be  investigated. 
We suspect that in finite samples, PPI may perform at least as well as the efficient estimators: for instance, when the sample size of the labeled data is small, or when not all of the features ($X$) in the labeled data are available at the individual level (e.g.,  due to privacy concerns). 
Second, our approach for accommodating covariate distribution shift required a sufficiently fast rate of estimation of the missing-data probability $\pi$. Alternate strategies,  such as entropy balancing \citep{chen2023entropy}, may avoid this requirement.

%% file: secAppendix.tex
\renewcommand{\thefigure}{S.\arabic{figure}}
\setcounter{figure}{0}
\newpage
\section*{Appendix}
\section{Proofs of main results}
\subsection{Proof of Theorem \ref{thm:recest}}
\label{sec:recestproof}
\begin{proof} 
Consider a rectified estimator with weight $\omega$ of the form
\begin{align*}
    \recomega{\theta}{\delta}=\labest{\theta}+\omega\rbr{\hat{\delta}_{n+N}-\labest{\delta}}.
\end{align*}
    From the asymptotic linearity of $\hat{\theta}_n$, $\hat{\delta}_n$, and $\hat{\delta}_{n+N}$, we have
    \begin{align*}
        \labest{\theta}-\target{\theta}&=\frac{1}{n}\sum_{i=1}^n\IFtheta(X_i,Y_i)+o_p\rbr{\frac{1}{\sqrt{n}}},\\
    \labest{\delta}-\target{\delta}&=\frac{1}{n}\sum_{i=1}^n\psi(X_i)+o_p\rbr{\frac{1}{\sqrt{n}}}, \text{ and }\\
    \hat{\delta}_{n+N}-\target{\delta}&=\frac{1}{n+N}\sum_{i=1}^{n+N}\psi(X_i)+o_p\rbr{\frac{1}{\sqrt{n+N}}}.
\end{align*} 
Thus, 
\begin{align*}
    \recomega{\theta}{\delta}-\target{\theta}&=\labest{\theta}+\omega\rbr{\hat{\delta}_{n+N}-\labest{\delta}}-\target{\theta}\\
    &=\frac{1}{n}\sum_{i=1}^n\IFtheta(X_i,Y_i)+\omega\sbr{\frac{1}{n+N}\rbr{\sum_{i=1}^n\IFdelta(X_i)+\sum_{i=n+1}^{n+N}\IFdelta(X_i)}-\frac{1}{n}\sum_{i=1}^n\IFdelta(X_i)}+o_p\rbr{\frac{1}{\sqrt{n}}}\\
    &=\frac{1}{n}\sum_{i=1}^n\IFtheta(X_i,Y_i)-\omega\frac{1}{1+\lambda}\frac{1}{n}\sum_{i=1}^n\IFdelta(X_i)+o_p(\frac{1}{\sqrt{n}})+\omega\frac{1}{1+\lambda}\frac{1}{N}\sum_{i=n+1}^{n+N}\IFdelta(X_i)+o_p\rbr{\frac{1}{\sqrt{N}}}\\
    &=\frac{1}{n}\sum_{i=1}^n\rbr{\IFtheta(X_i,Y_i)-\omega\frac{1}{1+\lambda}\IFdelta(X_i)}+o_p(\frac{1}{\sqrt{n}})+\frac{1}{N}\sum_{i=n+1}^{n+N}\omega\frac{1}{1+\lambda}\IFdelta(X_i)+o_p\rbr{\frac{1}{\sqrt{N}}}.
\end{align*}

By the central limit theorem, and the facts that $\bE[\IFtheta(X,Y)]=0$ and $\bE[\IFdelta(X)]=0$, we have that
\begin{align*}
    \sqrt{n}\rbr{\frac{1}{n}\sum_{i=1}^n\rbr{\IFtheta(X_i,Y_i)-\omega\frac{1}{1+\lambda}\IFdelta(X_i)}}&\xrightarrow{d}N(0,\sigma^2_1),\\
    \sqrt{N}\rbr{\frac{1}{N}\sum_{i=n+1}^{n+N}\omega\frac{1}{1+\lambda}\IFdelta(X_i)}&\xrightarrow{d}N(0,\sigma^2_2),
\end{align*}
where \begin{align*}
    &\sigma^2_1:=\Var\rbr{\IFtheta(X,Y)-\omega\frac{1}{1+\lambda}\IFdelta(X)}=\Var(\IFtheta(X,Y))-\frac{2\omega}{1+\lambda}\Cov(\IFtheta(X,Y),\IFdelta(X))+\frac{\omega^2}{(1+\lambda)^2}\Var(\IFdelta(X)),\\
    &\text{and }\sigma^2_2:=\Var\rbr{\omega\frac{1}{1+\lambda}\IFdelta(X)}=\frac{\omega^2}{(1+\lambda)^2}\Var(\IFdelta(X)).
\end{align*}
Therefore, by the independence of $\Lab_n$ and $\Unlab_N$, and recalling that $\frac{n}{N}\rightarrow\lambda$,
\begin{align*}
    \sqrt{n}\rbr{\recomega{\theta}{\delta}-\target{\theta}}&=\sqrt{n}\frac{1}{n}\sum_{i=1}^n\rbr{\IFtheta(X_i,Y_i)-\omega\frac{1}{1+\lambda}\IFdelta(X_i)}+\frac{\sqrt{n}}{\sqrt{N}}\sqrt{N}\frac{1}{N}\sum_{i=n+1}^{n+N}\omega\frac{1}{1+\lambda}\IFdelta(X_i)+o_p(1)\\
    &\xrightarrow{d}N\rbr{0,\Var(\IFtheta(X,Y))-\frac{2\omega}{1+\lambda}\Cov(\IFtheta(X,Y),\IFdelta(X))+\frac{\omega^2}{1+\lambda}\Var(\IFdelta(X))}.
\end{align*}

Now, recall the assumption that $\hat{\omega}\xrightarrow{p}\omega$, and the definition of $\recest{\theta}{\delta}$ in \eqref{eq:recest}. Then,
\begin{align*}
    \recest{\theta}{\delta}-\target{\theta}&=\labest{\theta}+\omega\rbr{\hat{\delta}_{n+N}-\labest{\delta}}-\target{\theta} + \rbr{\hat{\omega}-\omega}\rbr{\hat{\delta}_{n+N}-\labest{\delta}}\\
    &=\labest{\theta}+\omega\rbr{\hat{\delta}_{n+N}-\labest{\delta}}-\target{\theta} + \littleOp\rbr{\frac{1}{\sqrt{n}}}\\
    &=\recomega{\theta}{\delta}-\target{\theta}+\littleOp\rbr{\frac{1}{\sqrt{n}}},
\end{align*}
so it follows that 
\[
\sqrt{n}\rbr{\theta_{\hat{\omega}}^{\hat{\delta}}-\target{\theta}}\xrightarrow{d}N\rbr{0,\Var(\IFtheta(X,Y))-\frac{2\omega}{1+\lambda}\Cov(\IFtheta(X,Y),\IFdelta(X))+\frac{\omega^2}{1+\lambda}\Var(\IFdelta(X))}.
\]
The result in Theorem \ref{thm:recest} part (i) follows.

To prove (ii), we note that since the asymptotic variance $\rbr{\sigma_{\omega}^{\IFdelta}}^2$ in \eqref{eq:recest_var} is a quadratic function of $\omega$, it follows that $$\opt{\omega}:= \arg\min_{\omega}  \left\{ \bE[\IFtheta(X,Y)^2]-\frac{2\omega}{1+\lambda}\bE[\IFtheta(X,Y)\IFdelta(X)]+\frac{\omega^2}{1+\lambda}\bE[\IFdelta(X)^2] \right\}=\frac{\bE[\IFtheta(X,Y)\IFdelta(X)]}{\bE[\IFdelta(X)^2]}.$$
Plugging this into the expression for $\rbr{\sigma_{\omega}^{\IFdelta}}^2$ in \eqref{eq:recest_var} yields
$$  \rbr{\sigma_{\opt{\omega}}^{\IFdelta}}^2=\bE[\IFtheta(X,Y)^2]-\frac{\rbr{\bE[\IFtheta(X,Y)\IFdelta(X)]}^2}{(1+\lambda)\bE[\IFdelta(X)^2]},$$ 
which is  no greater than the asymptotic variance of $\hat{\theta}_n$, $\bE[\IFtheta(X,Y)^2]$.

Next, we turn to (iii). Note that minimizing $\sigma_{\opt{\omega}}^{\IFdelta}$
with respect to $\IFdelta$ is equivalent to maximizing $$\frac{\rbr{\bE[\IFtheta(X,Y)\IFdelta(X)]}^2}{\bE[\IFdelta(X)^2]}=\frac{\Cov(\IFtheta(X,Y),\IFdelta(X))^2}{\Var(\IFdelta(X))}$$ with respect to $\IFdelta$.

We define $\epsilon:= \IFtheta(X,Y)-\bE[\IFtheta(X,Y)\mid X]$. Then,
\begin{align*}
    \frac{\Cov(\IFtheta(X,Y),\IFdelta(X))^2}{\Var(\IFdelta(X))}&=\frac{\Cov(\bE[\IFtheta(X,Y)\mid X]+\epsilon,\IFdelta(X))^2}{\Var(\IFdelta(X))}\\
    &=\frac{\Cov(\bE[\IFtheta(X,Y)\mid X],\IFdelta(X))^2}{\Var(\IFdelta(X))}\\
    &\leq \Var(\bE[\IFtheta(X,Y)\mid X]),
\end{align*} where the final line follows from Cauchy-Schwarz and the second to last line follows because $\Cov(\epsilon,X)=0$.
The final line is an equality if and only if $\IFdelta(X)$ is linear in $\bE[\IFtheta(X,Y)\mid X]$.  Thus, $\rbr{\sigma_{\opt{\omega}}^{\IFdelta}}^2$ is minimized when 
 $\opt{\IFdelta}(X)=\bE[\IFtheta(X,Y)\mid X]$, and 
 $\rbr{\sigma_{\opt{\omega}}^{\opt{\IFdelta}}}^2=\bE[\IFtheta(X,Y)^2]-\frac{1}{1+\lambda}\Var(\bE[\IFtheta(X,Y)\mid X])$.
\end{proof}
\subsection{Proof of Proposition \ref{prop:PPest}}
\begin{proof}
    Under the regularity conditions described in the beginning of Section \ref{sec:rectifier}, $\labfest{\theta}$ and $\allfest{\theta}$ are regular ALEs  with influence function $x\rightarrow \IFtheta(x,f(x))$. Note that replacing the $\rbr{\hat{\delta}_{n+N}-\labest{\delta}}$ in \eqref{eq:recest} by $\rbr{\allfest{\theta}-\labfest{\theta}}$ gives $\ppest{\theta}$ in \eqref{eq:PPI}, and thus the results follow from (i) and (ii) in Theorem \ref{thm:recest}.
\end{proof}
\subsection{Proof of results in Proposition \ref{prop:PPest} with $\unlabfest{\theta}$ in the rectifier term}\label{sec:onlyN}
We here show that the theoretical results in Proposition \ref{prop:PPest} also hold for an alternative PPI estimator defined using $\unlabfest{\theta}$ (instead of $\allfest{\theta}$) in the rectifier term.
\begin{proposition}
    
\label{clm:only-N}
    Suppose that $\frac{n}{N}\rightarrow \lambda$ and that $\hat{\omega}_{\text{alt}}\xrightarrow{p}\omega_{\text{alt}}:=\frac{\bE[\IFtheta(X,Y)\IFtheta(X,f(X))]}{(1+\lambda)\bE[\IFtheta(X,f(X))^2]}$. Define $\ppest{\theta}_{\text{alt}}:=\labest{\theta}+\hat{\omega}_{\text{alt}}\rbr{\unlabfest{\theta}-\labfest{\theta}}.$ Then,
    \begin{equation*}
    \sqrt{n}\rbr{\ppest{\theta}_{\text{alt}}-\target{\theta}}\indist \Norm \rbr{0,\bE[\IFtheta(X,Y)^2]-\frac{\rbr{\bE[\IFtheta(X,Y)\IFtheta(X,f(X))]}^2}{(1+\lambda)\bE[\IFtheta(X,f(X))^2]}}.
    \end{equation*}
\end{proposition}

\begin{proof}
    From the asymptotic linearity of $\labest{\theta}$, $\labfest{\theta}$, and $\unlabfest{\theta}$, we have
    \begin{align*}
        \labest{\theta}-\target{\theta}&=\frac{1}{n}\sum_{i=1}^n\IFtheta(X_i,Y_i)+o_p\rbr{\frac{1}{\sqrt{n}}},\\
    \labfest{\theta}-\target{\theta}^f&=\frac{1}{n}\sum_{i=1}^n\IFtheta(X_i,f(X_i))+o_p\rbr{\frac{1}{\sqrt{n}}}, \text{ and }\\
    \unlabfest{\theta}-\target{\theta}^f&=\frac{1}{N}\sum_{i=1}^N\IFtheta(X_i,f(X_i))+o_p\rbr{\frac{1}{\sqrt{N}}}.
\end{align*} 
By the central limit theorem, and the facts that $\bE[\IFtheta(X,Y)]=0$, and $\bE[\IFtheta(X,f(X))]=0$, we have that
\begin{align*}
\sqrt{n}\left(\begin{pmatrix}
\labest{\theta}\\
\labfest{\theta}
\end{pmatrix}
-
\begin{pmatrix}
\target{\theta}\\
\target{\theta}^f
\end{pmatrix}\right)&\xrightarrow{d}N\left(\textbf{0},
\begin{pmatrix}
\bE[\IFtheta(X,Y)^2]  &\bE[\IFtheta(X,Y)\IFtheta(X,f(X))]\\
\bE[\IFtheta(X,Y)\IFtheta(X,f(X))] &\bE[\IFtheta(X,f(X))^2]
\end{pmatrix}\right),\\
\sqrt{N}\rbr{\unlabfest{\theta}-\target{\theta}^f}&\xrightarrow{d}N\rbr{0,\bE[\IFtheta(X,f(X))^2]}.
\end{align*}
Writing $h_\omega(a,b)=a-\omega b$ and $\nabla h_\omega=(1,-\omega)^\top$, by the delta method,
\[
\sqrt{n}\rbr{\labest{\theta}-\omega\labfest{\theta}-\target{\theta}+\omega\target{\theta}^f}\indist N(0,\bE[\IFtheta(X,Y)^2]-2\omega\bE[\IFtheta(X,Y)\IFtheta(X,f(X))]+\omega^2\bE[\IFtheta(X,f(X))^2]).
\]
Then, by the independence of $\Lab_n$ and $\Unlab_N$, and recalling that $\frac{n}{N}\rightarrow\lambda$, we have 
\begin{align*}
    &\sqrt{n}\rbr{\labest{\theta}+\omega(\unlabfest{\theta}-\labfest{\theta})-\target{\theta}}=\sqrt{n}\rbr{\labest{\theta}-\omega\labfest{\theta}-\target{\theta}+\omega\target{\theta}^f}+\frac{\sqrt{n}}{\sqrt{N}}\sqrt{N}\rbr{\omega\unlabfest{\theta}-\omega\target{\theta}^f}\\
    &\xrightarrow{d}N\rbr{0,\bE[\IFtheta(X,Y)^2]-2\omega\bE[\IFtheta(X,Y)\IFtheta(X,f(X))]+(1+\lambda)\omega^2\bE[\IFtheta(X,f(X))^2]}.
\end{align*}
Next, note that 
\begin{align*}
    \omega_{\text{alt}}&:= \frac{\bE[\IFtheta(X,Y)\IFtheta(X,f(X))]}{(1+\lambda)\bE[\IFtheta(X,f(X))^2]}\\
    &=\arg\min_{\omega}  \left\{ \bE[\IFtheta(X,Y)^2]-2\omega\bE[\IFtheta(X,Y)\IFtheta(X,f(X))]+(1+\lambda)\omega^2\bE[\IFtheta(X,f(X))^2]\right\}.
\end{align*}

Now, recall the assumption that $\hat{\omega}_{\text{alt}}\inprob \omega_{\text{alt}}$ and the definition of $\ppest{\theta}_{\text{alt}}$ in Claim \ref{clm:only-N}. Then,
\begin{align*}
    \ppest{\theta}_{\text{alt}}-\target{\theta}&=\labest{\theta}+\omega_{\text{alt}}\rbr{\unlabfest{\theta}-\labfest{\theta}}-\target{\theta} + \rbr{\hat{\omega}_{\text{alt}}-\omega_{\text{alt}}}\rbr{\unlabfest{\theta}-\labfest{\theta}}\\
    &=\labest{\theta}+\omega_{\text{alt}}\rbr{\unlabfest{\theta}-\labfest{\theta}}-\target{\theta}+\littleOp\left(\frac{1}{\sqrt{n}}\right).
\end{align*}
Plugging in $\omega_{\text{alt}}= \frac{\bE[\IFtheta(X,Y)\IFtheta(X,f(X))]}{(1+\lambda)\bE[\IFtheta(X,f(X))^2]}$, it follows that 
\begin{equation*}
    \sqrt{n}\rbr{\ppest{\theta}_{\text{alt}}-\target{\theta}}\indist \Norm \rbr{0,\bE[\IFtheta(X,Y)^2]-\frac{\rbr{\bE[\IFtheta(X,Y)\IFtheta(X,f(X))]}^2}{(1+\lambda)\bE[\IFtheta(X,f(X))^2]}}.
\end{equation*}
\end{proof}

\subsection{Proof of Proposition \ref{prop:aipw-f} }\label{sec:proofprop2.3}
To begin, we note that $C\indep Y, X$ and $\bE[\IFtheta(X,Y)]=0$ and $\bE[\IFtheta(X,f(X))]=0$. Therefore, we have $\bE\sbr{C\IFtheta(X,Y)}=0$, and $\bE\sbr{C\IFtheta(X,f(X))}=0$.

Next, recalling \eqref{eq:aipw}, we have that
\begin{align*}
    \mdppest{\theta}-\target{\theta}&=\frac{1}{n+N}\sum_{i=1}^{n+N}\rbr{\frac{C_i}{\hat{\pi}}\IFtheta(X_i,Y_i)+\hat{\omega}_{\text{MD}} \left(1-\frac{C_i}{\hat{\pi}}\right)\IFtheta(X_i,f(X_i))}\\
    &=\frac{1}{n+N}\sum_{i=1}^{n+N}\rbr{\frac{C_i}{\frac{\lambda}{1+\lambda}}\IFtheta(X_i,Y_i)+\omega\left(1-\frac{C_i}{\frac{\lambda}{1+\lambda}}\right)\IFtheta(X_i,f(X_i))}+o_p\rbr{\frac{1}{\sqrt{n+N}}},
\end{align*}
where the second equality follows from noting that $\hat{\pi}\xrightarrow{p}\Pr(C=1)=\frac{\lambda}{1+\lambda}$.

Therefore,
\[
\sqrt{n+N}\rbr{\mdppest{\theta}-\target{\theta}}\xrightarrow{d}N\rbr{0,\sigma^2_{\text{MD}}},
\]where $\sigma^2_{\text{MD}}=\frac{1+\lambda}{\lambda}\bE\sbr{\IFtheta(X,Y)^2}-\frac{2\omega}{\lambda}\bE\sbr{\IFtheta(X,Y)\IFtheta(X,f(X))}+\frac{\omega^2}{\lambda}\bE\sbr{\IFtheta(X,f(X))^2}$.

Rescaling,
\[
\sqrt{n}\rbr{\mdppest{\theta}-\target{\theta}}\xrightarrow{d}N\rbr{0,\bE\sbr{\IFtheta(X,Y)^2}-\frac{2\omega}{1+\lambda}\bE\sbr{\IFtheta(X,Y)\IFtheta(X,f(X))}+\frac{\omega^2}{1+\lambda}\bE\sbr{\IFtheta(X,f(X))^2}}
\] matches the asymptotic results for $\ppest{\theta}$ \eqref{eq:PPI}.

\subsection{Proofs of Remarks \ref{rem:full_fdm}, \ref{rem:unlab_fdm}, and \ref{rem:lab_fdm} }\label{app:fdm-ale}

\subsubsection{Proof of Remark \ref{rem:full_fdm}}\label{sec:appremfull}
\begin{proof}
Recall that $(X_i,Y_i)\iidsim\jointcov_{X,Y}$, 
$\target{\theta}:=\Phi(\target{F})$,
 where $\target{F}(x,y)=\jointcov_{X,Y}(X\leq x,Y\leq y)$. \\
Define $F_\labfull(x,y):=\frac{\sum_{i=1}^{n+N}\frac{C_i}{\pi(X_i)}I(X_i\leq x,Y_i\leq y)}{\sum_{i=1}^{n+N}\frac{C_i}{\pi(X_i)}}$. Then,
\begin{align*}
    F_\labfull(x,y)-\target{F}(x,y)&=\frac{\sum_{i=1}^{n+N}\frac{C_i}{\pi(X_i)}\rbr{I(X_i\leq x, Y_i \leq y)-\target{F}(x,y)}}{\sum_{i=1}^{n+N}\frac{C_i}{\pi(X_i)}}\\
    &=\frac{\frac{1}{n+N}\sum_{i=1}^{n+N}\frac{C_i}{\pi(X_i)}(I(X_i\leq x, Y_i \leq y)-\target{F}(x,y))}{\frac{1}{n+N}\sum_{i=1}^{n+N}\frac{C_i}{\pi(X_i)}}\\
    &=\frac{1}{n+N}\sum_{i=1}^{n+N}\frac{C_i}{\pi(X_i)}(I(X_i\leq x, Y_i \leq y)-\target{F}(x,y))\\
    &\qquad+\rbr{\frac{1}{\frac{1}{n+N}\sum_{i=1}^{n+N}\frac{C_i}{\pi(X_i)}}-1}\frac{1}{n+N}\sum_{i=1}^{n+N}\frac{C_i}{\pi(X_i)}(I(X_i\leq x, Y_i \leq y)-\target{F}(x,y))\\
    &=\frac{1}{n+N}\sum_{i=1}^{n+N}\frac{C_i}{\pi(X_i)}(I(X_i\leq x, Y_i \leq y)-\target{F}(x,y))+\littleOp\rbr{\frac{1}{\sqrt{n+N}}}.
\end{align*}
Note that given $C\indep Y\mid X$, it follows that $\bE[\frac{C}{\pi(X)}(I(X\leq x, Y \leq y)-\target{F}(x,y))]=0$. 

Suppose $\Phi(\cdot)$ is Hadamard differentiable with respect to $||\cdot||_\infty$ with Gateaux derivative $\dot{\Phi}(\target{F};h)$ that is linear in $h$. By the functional delta method,
\begin{align*}
\Phi\rbr{F_\labfull}-\Phi(\target{F})&=\dot{\Phi}\rbr{\target{F};\frac{1}{n+N}\sum_{i=1}^{n+N}\frac{C_i}{\pi(X_i)}\rbr{I(X_i\leq x, Y_i\leq y)-\target{F}(x,y)}}+\littleOp\rbr{\frac{1}{\sqrt{n+N}}}\\
&=\frac{1}{n+N}\sum_{i=1}^{n+N}\frac{C_i}{\pi(X_i)}\dot{\Phi}(\target{F};I(X_i\leq x, Y_i\leq y)-\target{F}(x,y))+\littleOp\rbr{\frac{1}{\sqrt{n+N}}}\\
&=\frac{1}{n+N}\sum_{i=1}^{n+N}\frac{C_i}{\pi(X_i)}\IFtheta(X_i,Y_i)+\littleOp\rbr{\frac{1}{\sqrt{n+N}}},
\end{align*}
where the last line follows from the definition in \eqref{eq:generalif_fdm}. Therefore, $\fulllabmdest{\theta}:=\Phi\rbr{F_\labfull}$ has influence function $\varphi_\lab(x,y,c)=\frac{c}{\pi(x)}\IFtheta(x,y)$.

Following similar arguments, one can show $\fulllabmdestf{\theta}:=\Phi\rbr{F^{f}_\labfull}$ is a regular ALE with  influence function $\varphi_\lab(x,f(x),c)=\frac{c}{\pi(x)}\IFtheta(x,f(x))$; and $\fullmdestf{\theta}:=\Phi\rbr{F^{f}}$ is a regular ALE  with influence function evaluated $\varphi(x,f(x),c):=\IFtheta(x,f(x))$.
\end{proof}

\subsubsection{Proofs of Remarks \ref{rem:unlab_fdm} and \ref{rem:lab_fdm} }
The proofs for these two remarks follow exactly as shown in Section \ref{sec:appremfull} but with weights adjusted according to the estimators defined in Sections \ref{sec:covdistunlab} and \ref{sec:covdistlab}.
\subsection{Proofs of Theorem \ref{thm:ppiipw}, Theorem \ref{thm:ppiiow}, and Theorem \ref{thm:ppiow}}\label{app:covshift}
Based on Remarks \ref{rem:full_fdm},  \ref{rem:unlab_fdm}, and \ref{rem:lab_fdm}, estimators defined in \eqref{eq:estfull}, \eqref{eq:unlabest}, and \eqref{eq:labest} are regular ALEs. The proofs follow the same logic as in Section  \ref{sec:recestproof}.

\subsection{Proof of Remark \ref{prop:covshiftestpifull}} \label{app:proofrem4.1}
Let $\bQ_{n+N}$ denote the empirical law of $(C_i,X_i,Y_i)_{i=1}^{n+N}$, and let $\jointcov$ denote the true law. Recall that $T:\sQ \times \Pi\rightarrow \bR$ is jointly Hadamard differentiable tangentially to $L_0^2(\jointcov)\times L^2(\jointcov_X)$ at $(\jointcov,\pi)$, so that $\target{\theta}$, $\fulllabmdest{\theta}$, and $\fulllabmdestpi{\theta}$ can be viewed as derived from the mapping $T$. That is,
\begin{align*}
\target{\theta}&=T(\jointcov,\pi),\\
    \fulllabmdest{\theta}&= T(\bQ_{n+N},\pi),\\
    \fulllabmdestpi{\theta} &= T(\bQ_{n+N},\hat{\pi}).
\end{align*}
Furthermore, from Remark \ref{rem:full_fdm} we know $\fulllabmdest{\theta}$ is a regular ALE of $\target{\theta}$ with the influence function evaluated at a single observation $\frac{C_i}{\pi(X_i)}\IFtheta(X_i,Y_i)$. Therefore, we have that
\begin{align*}
    T(\bQ_{n+N},\pi)-T(\jointcov,\pi)&=\fulllabmdest{\theta}-\target{\theta}\\
    &=\frac{1}{n+N}\sum_{i=1}^{n+N}\frac{C_i}{\pi(X_i)}\IFtheta(X_i,Y_i)+o_p(1/\sqrt{n+N}).
\end{align*}


Define $\bQ_{\epsilon_1}=\jointcov+\epsilon_1(\bQ_{n+N}-\jointcov)$ and $\pi_{\epsilon_2}=\pi+\epsilon_2(\hat{\pi}-\pi)$. Then, by joint Hadamard differentiability,
\begin{align*}
    \labmdestpi{\theta} - \target{\theta} &= T(\bQ_{n+N},\hat{\pi})-T(\jointcov,\pi)\\
    &=\frac{d}{d\epsilon_1}T(\bQ_{\epsilon_1},\pi)\bigg|_{\epsilon_1=0}+\frac{d}{d\epsilon_2}T(\jointcov,\pi_{\epsilon_2})\bigg|_{\epsilon_2=0}+o\rbr{||\bQ_{n+N}-\jointcov||_{L_0^2(\jointcov)}+||\hat{\pi}-\pi||_{L^2(\jointcov_X)}}.
\end{align*}
Algebraic manipulations and Assumption \ref{assmp:estpi} then yield that
\begin{align*}
\labmdestpi{\theta} - \target{\theta} &=\frac{1}{n+N}\sum_{i=1}^{n+N}\varphi_{\text{lab}}(C_i,X_i,Y_i)+o_p(1/\sqrt{n+N}),
\end{align*}
where $\varphi_{\text{lab}}(c,x,y)=\frac{c}{\pi(x)}\phi(x,y)-\bE_{\jointcov}\sbr{\frac{C\phi(X,Y)}{\pi(X)^2}\phi_\pi(c,x;X)}$.

Define $\bQ_{n+N}^f$ as the empirical law of $(C_i,X_i,f(X_i))_{i=1}^{n+N}$ and thus $\labmdestpif{\theta}=T(\bQ_{n+N}^f,\hat{\pi})$. A similar argument establishes that under regularity conditions, $\labmdestpif{\theta}$ is a regular ALE with the influence function $\varphi_{\text{lab}}(c,x,f(x))=\frac{c}{\pi(x)}\phi(x,f(x))-\bE_{\jointcov}\sbr{\frac{C\phi(X,f(X))}{\pi(X)^2}\phi_\pi(c,x;X)}$.

\section{Supplementary results}
\subsection{Additional discussion of theoretical results}
\subsubsection{PPI with covariate distribution shift for general regular ALEs}\label{sec:covdistgeneral}

We showed in in Remarks \ref{rem:full_fdm}, \ref{rem:unlab_fdm}, and \ref{rem:lab_fdm} that the components of the PPI estimators in Sections \ref{sec:covdistall}, \ref{sec:covdistunlab}, and \ref{sec:covdistlab} all have influence functions that are some linear transformation of $\IFtheta(x,y)$ in \eqref{eq:generalif_fdm}. 

Here, in the scenario of Section \ref{sec:covdistall}, we show that a result similar to Theorem~\ref{thm:ppiipw} can be obtained regardless of the form of the influence functions of the components of the PPI estimator.   Similar results can be derived for the scenarios in Sections \ref{sec:covdistunlab} and \ref{sec:covdistlab}.

Recall that our PPI estimator in Section \ref{sec:covdistall} is defined as 
\begin{equation}\label{eq:estfullgeneral}
    \fullest{\theta}:= \fulllabmdest{\theta}+\hat{\omega}\rbr{\fullmdestf{\theta}-\fulllabmdestf{\theta}}.
\end{equation}
However, $\fulllabmdest{\theta}$, $\fullmdestf{\theta}$, and $\fulllabmdestf{\theta}$ do not have to follow the definitions in Section  \ref{sec:covdistall}. We instead assume only that $\fulllabmdest{\theta}$, $\fulllabmdestf{\theta}$, and $\fullmdestf{\theta}$ are regular ALEs with influence functions $\varphi_\lab(X_i,Y_i,C_i)$, $\varphi_\lab(X_i,f(X_i),C_i)$, and $\varphi(X_i,f(X_i),C_i)$. Thus, the asymptotic variance of $\fulllabmdest{\theta}$ is
$\Var\rbr{\varphi_\lab(X_i,Y_i,C_i)}$. The next result shows that the asymptotic variance of $\fullest{\theta}$ \eqref{eq:estfullgeneral} will never exceed that.

\begin{proposition}[PPI for targets of the full data distribution for more general regular ALEs]\label{thm:ppiipwgeneral}
Under Assumption \ref{assmp:covdistshift} and the further assumption that  $\hat{\omega}\inprob \omega$, where
    $$\omega:=\frac{\Cov\Big(\varphi_\lab(X_i,Y_i,C_i),\varphi_\lab(X_i,f(X_i),C_i)-\varphi(X_i,f(X_i),C_i)\Big)}{\Var\Big(\varphi(X_i,f(X_i),C_i)-\varphi_\lab(X_i,f(X_i),C_i)\Big)},$$  it follows that
$$\sqrt{n+N}\rbr{\fullest{\theta}-\target{\theta}}\indist N(0,\sigma^2),$$ where $\sigma^2=\Var\left(\varphi_\lab(X_i,Y_i,C_i)\right)-\frac{\Cov\left(\varphi_\lab(X_i,Y_i,C_i),\varphi_\lab(X_i,f(X_i),C_i)-\varphi(X_i,f(X_i),C_i)\right)^2}{\Var\left(\varphi(X_i,f(X_i),C_i)-\varphi_\lab(X_i,f(X_i),C_i)\right)}$.
\end{proposition}

\subsubsection{Variance comparison to \cite{kluger2025prediction} under Assumption \ref{assmp:oracle_case}}\label{sec:ptd_comparison}
The variance of the tuned Predict-Then-Debias (TPTD) estimator proposed by \cite{kluger2025prediction} (in their Proposition 2.2) in the setting of Section \ref{sec:covdistall} is 
\begin{equation*}
    \sigma^2_{\text{TPTD}}=\Var\rbr{\frac{C}{\pi(X)}\phi(X,Y)}-\frac{\Cov\rbr{\frac{C}{\pi(X)}\phi(X,Y),\frac{C}{\pi(X)}\phi(X,f(X))}^2}{\Var\rbr{\frac{C}{\pi(X)}\phi(X,f(X))}+\Var\rbr{\frac{1-C}{1-\pi(X)}\IFtheta(X,f(X))}}.
\end{equation*}
Recall that under the (highly unrealistic) Assumption \ref{assmp:oracle_case}, we have that  $\phi(X,f(X))=\bE[\phi(X,Y)\mid X]$. In this setting, 
\begin{equation}\label{eq:ptdvar}
    \sigma^2_{\text{TPTD}}=\Var\rbr{\frac{C}{\pi(X)}\phi(X,Y)}-\frac{\rbr{\bE\sbr{\frac{1}{\pi(X)}\rbr{\bE[\phi(X,Y)\mid X]}^2}}^2}{\bE\sbr{\frac{1}{\pi(X)(1-\pi(X))}\rbr{\bE[\phi(X,Y)\mid X]}^2}}.
\end{equation}
Furthermore, under Assumption \ref{assmp:oracle_case}, the variance of our PPI estimator is (see Equation \eqref{eq:oracle_full}) \[\sigma^2=\Var\left(\frac{C}{\pi(X)}\phi(X,Y)\right)-\bE\sbr{\rbr{\frac{1}{\pi(X)}-1}\bE[\phi(X,Y)\mid X]^2}. \]
To compare $\sigma^2_{\text{TPTD}}$ \eqref{eq:ptdvar} and $\sigma^2$ in \eqref{eq:oracle_full}, we only need to compare $\frac{\rbr{\bE\sbr{\frac{1}{\pi(X)}\rbr{\bE[\phi(X,Y)\mid X]}^2}}^2}{\bE\sbr{\frac{1}{\pi(X)(1-\pi(X))}\rbr{\bE[\phi(X,Y)\mid X]}^2}}$ and $\bE\sbr{\rbr{\frac{1}{\pi(X_i)}-1}\bE[\phi(X_i,Y_i)\mid X_i]^2}.$

Note that\[
\frac{1}{\pi(X)}\rbr{\bE[\phi(X,Y)\mid X]}^2=\sqrt{\frac{1}{\pi(X)(1-\pi(X))}\rbr{\bE[\phi(X,Y)\mid X]}^2}  \sqrt{\frac{1-\pi(X)}{\pi(X)}\rbr{\bE[\phi(X,Y)\mid X]}^2}.
\]
By Cauchy-Schwarz, 
\begin{align*}
    &\rbr{\bE\sbr{\frac{1}{\pi(X)}\rbr{\bE[\phi(X,Y)\mid X]}^2}}^2\\
    &=\rbr{\bE\sbr{\sqrt{\frac{1}{\pi(X)(1-\pi(X))}\rbr{\bE[\phi(X,Y)\mid X]}^2} \sqrt{\frac{1-\pi(X)}{\pi(X)}\rbr{\bE[\phi(X,Y)\mid X]}^2}}}^2 \nonumber\\
    &\leq \bE\sbr{\frac{1}{\pi(X)(1-\pi(X))}\rbr{\bE[\phi(X,Y)\mid X]}^2}\bE\sbr{\frac{1-\pi(X)}{\pi(X)}\rbr{\bE[\phi(X,Y)\mid X]}^2}. 
\end{align*}
Provided $\bE[\phi(X,Y)\mid X]\neq 0$,  equality requires that the following holds for some constant $b$:
\begin{align*}
    \sqrt{\frac{1}{\pi(X)(1-\pi(X))}\rbr{\bE[\phi(X,Y)\mid X]}^2}&=b\sqrt{\frac{1-\pi(X)}{\pi(X)}\rbr{\bE[\phi(X,Y)\mid X]}^2} \\
    \Leftrightarrow \frac{1}{\pi(X)(1-\pi(X))} &=b^2\frac{1-\pi(X)}{\pi(X)}\\
    \Leftrightarrow 1-\pi(X)&=\sqrt{\frac{1}{b^2}}.
\end{align*}
This condition requires $\pi(X)$ to be a constant, which violates the setup in Section \ref{sec:covdistall}. Therefore, when $\pi(X)$ is not a constant, we have the strict inequality,
\begin{align}\label{eq:ineqtptd}
    &\rbr{\bE\sbr{\frac{1}{\pi(X)}\rbr{\bE[\phi(X,Y)\mid X]}^2}}^2\\ \nonumber
    &< \bE\sbr{\frac{1}{\pi(X)(1-\pi(X))}\rbr{\bE[\phi(X,Y)\mid X]}^2}\bE\sbr{\frac{1-\pi(X)}{\pi(X)}\rbr{\bE[\phi(X,Y)\mid X]}^2}.
\end{align}

We now return to the comparison of \eqref{eq:ptdvar} and \eqref{eq:oracle_full}. 

Dividing
both sides of \eqref{eq:ineqtptd} by $\bE\sbr{\frac{1}{\pi(X)(1-\pi(X))}\rbr{\bE[\phi(X,Y)\mid X]}^2}$, we see that
\[
\frac{\rbr{\bE\sbr{\frac{1}{\pi(X)}\rbr{\bE[\phi(X,Y)\mid X]}^2}}^2}{\bE\sbr{\frac{1}{\pi(X)(1-\pi(X))}\rbr{\bE[\phi(X,Y)\mid X]}^2}} < \bE\sbr{\rbr{\frac{1}{\pi(X_i)}-1}\bE[\phi(X_i,Y_i)\mid X_i]^2}.
\]
Therefore,  $\sigma^2_{\text{TPTD}}$ \eqref{eq:ptdvar}  is greater than $\sigma^2$ \eqref{eq:oracle_full}.

\subsubsection{Comparison to debiased AIPW in the setting of Section \ref{sec:estpi}}\label{sec:appcomaipw}
The extra first-order term in the influence function (in Remark \ref{prop:covshiftestpifull}) introduced by estimating $\pi$ complicates the derivation of the influence function, and thereby the estimation of the variance. A natural question is: is it possible for the influence function to have a simple form while using a simple PPI estimator of the form in \eqref{eq:ppiestpi} with estimated $\pi$? The answer is yes, but there is a trade-off. 

To eliminate the extra first-order term in the influence function in Remark \ref{prop:covshiftestpifull}, we can impose assumptions on the quality of $f(X)$. Therefore,  we will view it as $\hat{f}_M(X)$, to emphasize that it is trained on $M$ external observations that are independent of the $n + N$ observations used for inference. Instead of Assumption \ref{assmp:estpi}, we assume
\begin{itemize}
    \item  $||\hat{\pi}-\pi||_{L^2(\jointcov)}\rightarrow 0$ as $n+N \rightarrow \infty$, $||\IFtheta(X,\hat{f}_M(X))-\bE[\IFtheta(X,Y)\mid X]||_{L^2(\jointcov)}\rightarrow 0$ as $M \rightarrow \infty$,
    \item  and $||\hat{\pi}-\pi||_{L^2(\jointcov)}\cdot||\IFtheta(X,\hat{f}_M(X))-\bE[\IFtheta(X,Y)\mid X]||_{L^2(\jointcov)}=o_p\rbr{\frac{1}{\sqrt{n+N}}}$.
\end{itemize}

Then, in this case, provided that $\IFtheta(x,y)$ is the efficient influence function, the asymptotic variance in Theorem \ref{thm:covshiftallestpi} will reduce to \eqref{eq:oracle_full}, achieving the asymptotic semiparametric efficiency bound. Compared to Assumption \ref{assmp:estpi}, these are slightly weaker assumptions on the quality of $\hat{\pi}$, but are stronger assumptions on the quality of $f$, and they enable gains in efficiency and neatness in the characterization of the influence function (eliminating the extra term we had in Remark \ref{prop:covshiftestpifull}).

Instead of relying on $\IFtheta(X,\hat{f}_M(X))$ as a sufficiently good estimator of $\bE[\IFtheta(X,Y)\mid X]$, suppose that one is willing to pursue a more refined approach, by appropriately estimating $\bE[\IFtheta(X,Y)\mid X]$ at a desired rate. Then, one can construct the AIPW estimator with this estimate, $\hat{\bE}[\IFtheta(X,Y)\mid X]$ \citep{chernozhukov2018double}. This will also achieve the semiparametric efficiency bound \citep{robins1994estimation,robins1995semiparametric,tsiatis2006semiparametric}, provided that $\IFtheta(x,y)$ is the efficient influence function.

In a practical finite-sample setting, we believe that the approach in Section \ref{sec:estpi} provides a good balance between ease of use and performance.

\subsection{Additional simulation results}\label{sec:xtrasims}
\subsubsection{Variance decrease in TPR and FPR estimation with noisy predictors}\label{sec:reducintprfpr}
From Section \ref{sec:rectifier}, we know the variance decrease arising from PPI depends on $\Cov(\IFtheta(X,Y), \IFtheta(X,f(X)))$. This value depends on the actual form of $\IFtheta$. From Corollary \ref{cor:auc_ppi},  $\Cov(\IFtheta^\AUC(X,Y), \IFtheta^\AUC(X,f(X)))$ can be simplified to $c_\AUC\Cov(Y,f(X))$, where $c_\AUC$ is some constant. Therefore, when $\Cov(f(X),Y)=0$, the PPI estimator does not improve AUC estimation in terms of asymptotic variance.
By contrast, based on the forms of $\IFtheta^\TPR$ and $\IFtheta^\FPR$ in 
Corollaries \ref{cor:tpr_ppi} and \ref{cor:fpr_ppi}, neither  $\Cov(\IFtheta^\TPR(X,Y),\\ \IFtheta^\TPR(X,f(X)))$ nor $\Cov(\IFtheta^\FPR(X,Y), \IFtheta^\FPR(X,f(X)))$ equals 0 when $\Cov(f(X),Y)=0$. Hence, even if $f(X)$ has no association with $Y$, it still leads to variance reduction in TPR and FPR estimation.

\subsubsection{Simulations illustrating Section \ref{rem:fequalsR}}\label{sec:simforrem}
We now consider the setting of Section \ref{rem:fequalsR}, where $f(X)=R$. This setting could arise when the goal is to evaluate TPR, FPR, and AUC for a prediction model $f$. We will now numerically show that PPI estimators  have smaller variances than estimators that use only the labeled data, provided that $\Cov(\IFtheta(X,Y),\IFtheta(X,f(X)))\neq 0$.

We generate data as described in Section \ref{sec:datagen}. Since $f(X)=R$, the targets are 
\begin{align*}
    \target{\theta}^\TPR &:= \Pr(f(\textbf{X})>\alpha\mid Y=1)=\frac{\bE[I(f(\textbf{X})>\alpha)Y]}{\bE[Y]},  \\
    \target{\theta}^\FPR &:= \Pr(f(\textbf{X})>\alpha\mid Y=0)=\frac{\bE[I(f(\textbf{X})>\alpha)(1-Y)]}{\bE[(1-Y)]}, \text{\;\;\;and}  \\
    \target{\theta}^\AUC &:= \int \Pr(f(\textbf{X})>\alpha\mid Y=1)d\Pr(f(\textbf{X})\leq\alpha\mid Y=0).
\end{align*}

We investigate two forms of $f(X)$, as follows:
\begin{itemize}
    \item $\Pr(Y=1\mid \textbf{X})$: This is the ideal model, i.e., $f(\textbf{X})=\bE[Y\mid \textbf{X}]$. Based on the form of the influence functions shown in Section \ref{sec:binary_no_covshift}, Assumption \ref{assmp:oracle_case} is satisfied for estimating TPR, FPR, and AUC.
    \item $\text{RF}(Y\sim \textbf{X})$: A random forest model predicting $Y$ from $\textbf{X}$.
\end{itemize}
These prediction models are independent of the data for conducting inference. The random forest model is trained on 100,000 independent observations. We set $n=1000$ and $\lambda=0.1$. We only generate outcomes $Y$ for the $n$ observations, and generate $f(\textbf{X})$ for all $n+N$ observations. The simulation results are evaluated over 2500 replications.

\begin{figure}[htbp]
  \centering

  \begin{subcaptionbox}{$f(\textbf{X}):\Pr(Y=1\mid \textbf{X})$\label{fig:rem_ideal_ratio}}[0.48\textwidth]
    {\includegraphics[width=\linewidth]{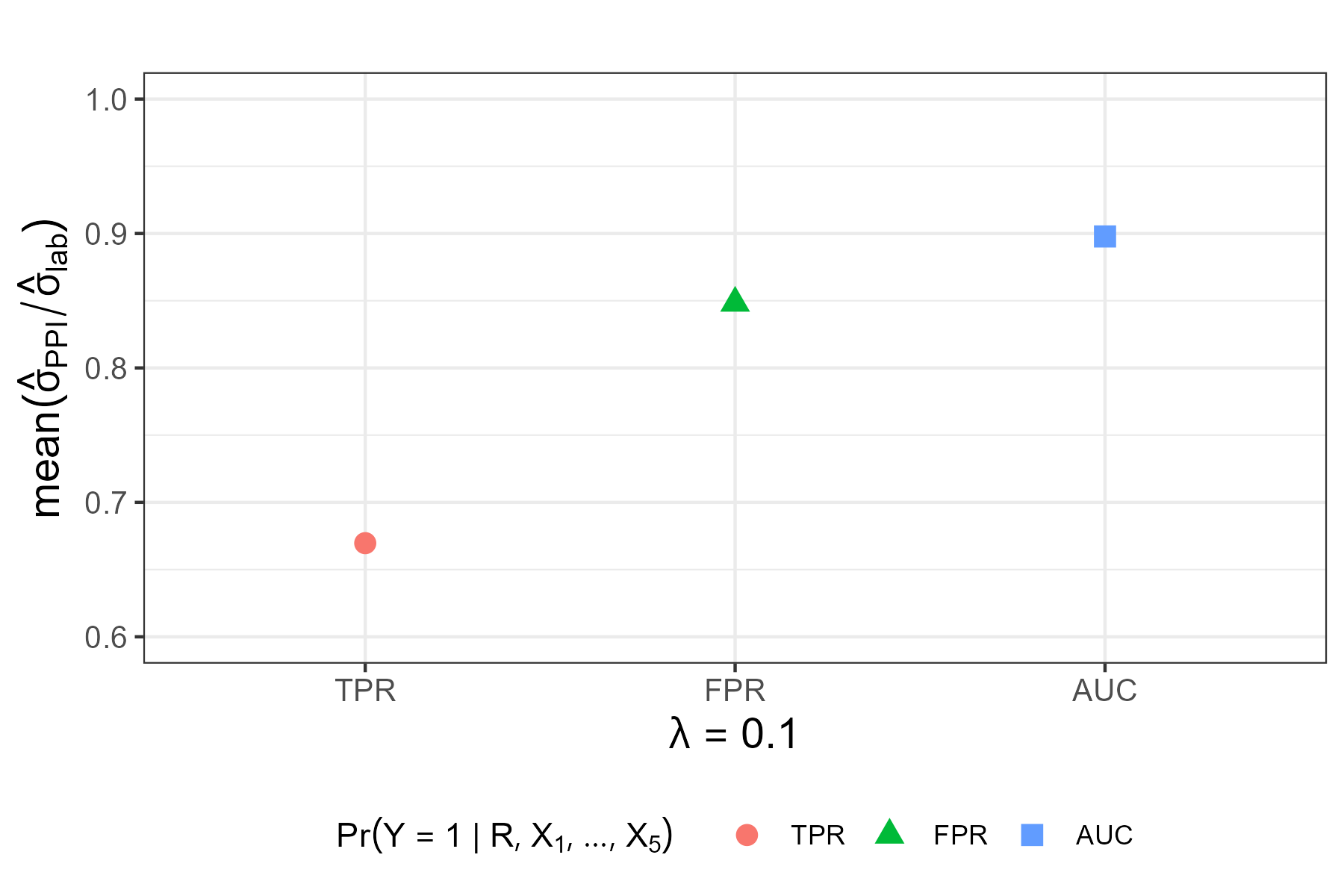}}
  \end{subcaptionbox}
  \begin{subcaptionbox}{$f(\textbf{X}):\text{RF}(Y\sim \textbf{X})$\label{fig:rem_all_ratio}}[0.48\textwidth]
    {\includegraphics[width=\linewidth]{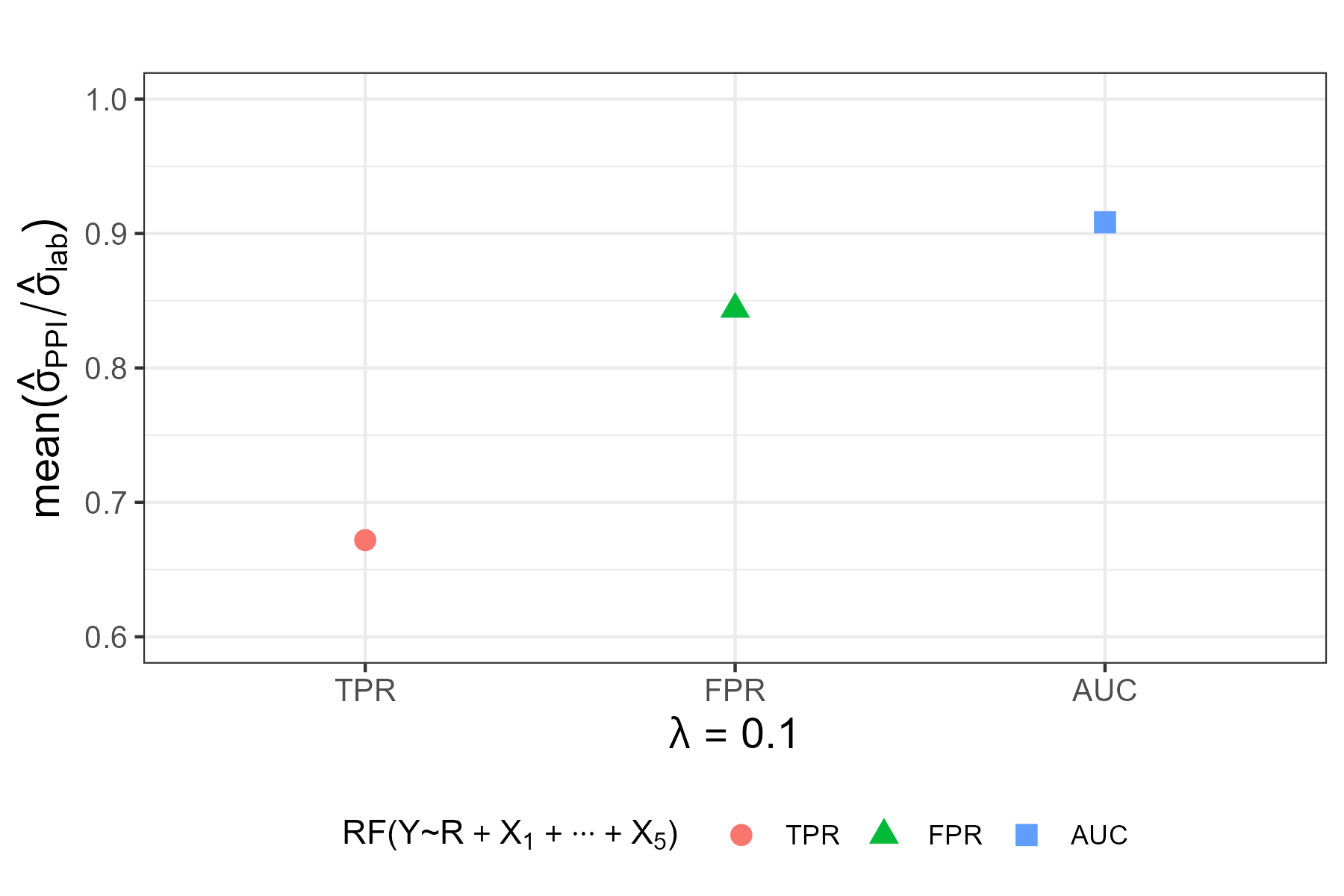}}
  \end{subcaptionbox}
  \caption{Ratio of estimated standard error for the PPI estimator to that of the estimator using only labeled data, averaged over 2500 replications, with $n=1000$, $\lambda=0.1$, for TPR, FPR, and AUC estimation (shapes), in a setting where $R=f(X)$.}
  \label{fig:rem_ratio}
\end{figure}

Figure \ref{fig:rem_ratio} compares the estimated standard errors of the PPI estimators to those of the estimators using only the labeled data for TPR, FPR, and AUC estimation. TPR and FPR are evaluated at a threshold value of 0.6. 
For both $f(\textbf{X})=\Pr(Y=1\mid \textbf{X})$ and $f(\textbf{X})=\text{RF}(Y\sim \textbf{X})$, none of $\Cov(\IFtheta^\TPR(f(\textbf{X}),Y),\IFtheta^\TPR(f(\textbf{X}),f(\textbf{X})))$, $\Cov(\IFtheta^\FPR(f(\textbf{X}),Y),\IFtheta^\FPR(f(\textbf{X}),f(\textbf{X})))$, or $\Cov(\IFtheta^\AUC(f(\textbf{X}),Y),\IFtheta^\AUC(f(\textbf{X}),f(\textbf{X})))$ equals 0. The simulation results align with our theoretical results in Remark \ref{rem:fequalsR}:   the variances of the PPI estimators are smaller than those of the estimators using just the labeled data. 
\begin{figure}[htbp]
    \centering
    \includegraphics[width=\linewidth]{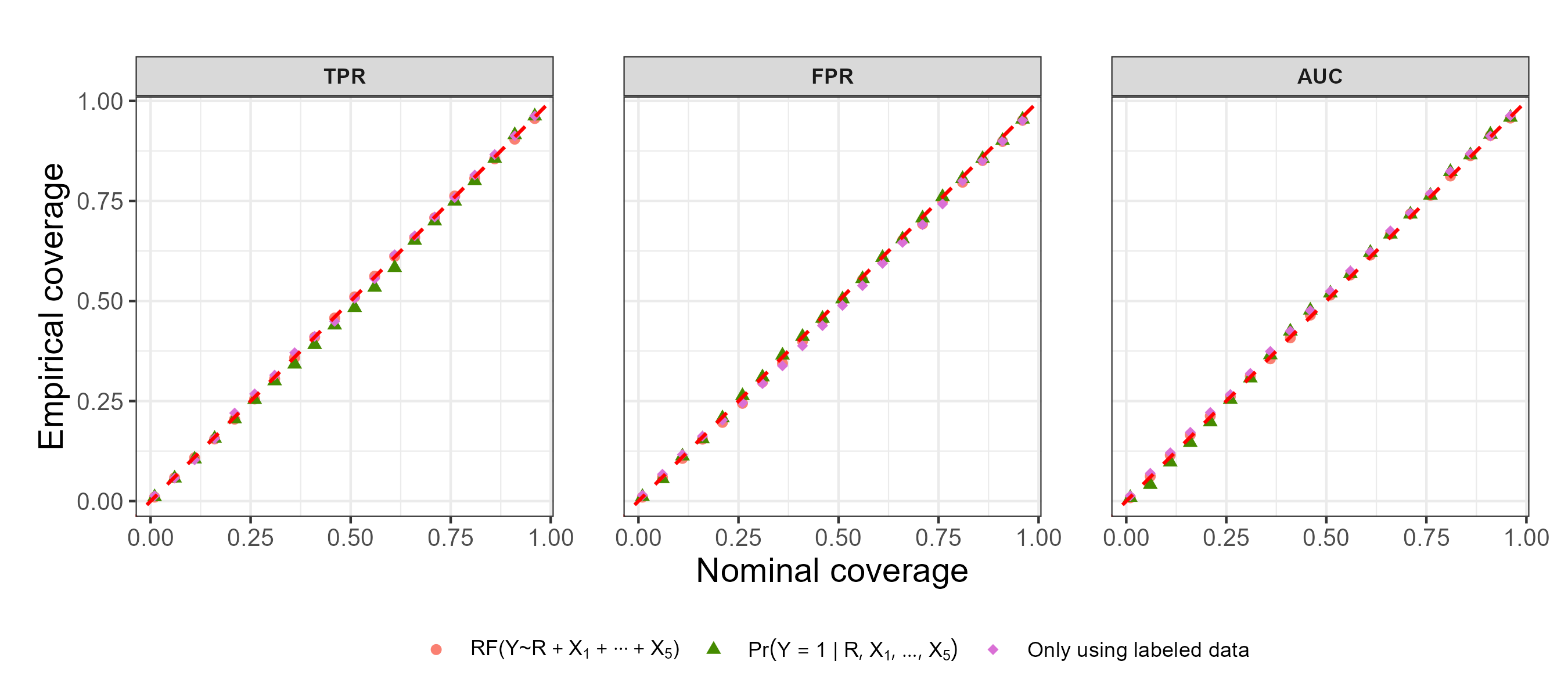}
    \caption{Empirical coverage versus nominal coverage, averaged over 2500 replications, with $n=1000$, $\lambda=0.1$, for TPR, FPR, and AUC estimation, in a setting where $f(\textbf{X})=R$. Shapes denote different prediction models $f$ used in PPI. The red dashed line is the diagonal line. The empirical coverages are close to the nominal coverages.}
    \label{fig:rem_cov}
\end{figure}
Figure \ref{fig:rem_cov} shows that the empirical coverages are close to the nominal coverages.

\subsubsection{Simulations results under a different data generation process}\label{sec:cordgp}
We now generate data with covariates $\textbf{X}:=\rbr{R,X_1,X_2}^\text{T}\in \bR^3$ and outcome $Y \in \cbr{0,1}$ following\\
$\begin{pmatrix}
R  \\
X_1\\
X_2
\end{pmatrix}\iidsim\text{MVN}\rbr{\begin{pmatrix}
0 \\
2\\
0
\end{pmatrix},\Sigma}$, where $\Sigma=\begin{bmatrix}
1 & 0.9 & 0.82 \\
0.9 & 1 & 0.49  \\
0.82 & 0.49 & 1 
\end{bmatrix}$, 
$Y\mid \textbf{X}\overset{\text{ind}}{\sim} \text{Bernoulli}(\Pr(Y=1\mid \textbf{X}))$, and
logit$(\Pr(Y=1\mid \textbf{X}))=R$.

We investigate three candidate prediction models for $f$, as follows:
\begin{itemize}
    \item $\Pr(Y=1\mid \textbf{X})$: This is the ideal model, i.e., $f(\textbf{X})=\bE[Y\mid \textbf{X}]$. Based on the form of the influence functions shown in Section \ref{sec:binary_no_covshift}, Assumption \ref{assmp:oracle_case} is satisfied for estimating TPR, FPR, and AUC, as well as for the mean.
    \item $\text{RF}(Y\sim X_1+X_2)$: A random forest model predicting $Y$ from $X_1$ and $X_2$.
    \item $\text{Unif}[0.01,0.99]$: Uniformly distributed random noise. 
\end{itemize}
These prediction models are independent of the data for conducting inference. Specifically, the random forest models are trained on 100,000 independent observations.

We set $n=1000$ and $\lambda=0.1$. We only generate the outcomes $Y$ for $n$ observations. We generate $\textbf{X}$ over all $n+N$ observations. The simulation results are evaluated over 2500 replications.

\begin{figure}[ht]
    \centering
    \includegraphics[width=1\linewidth]{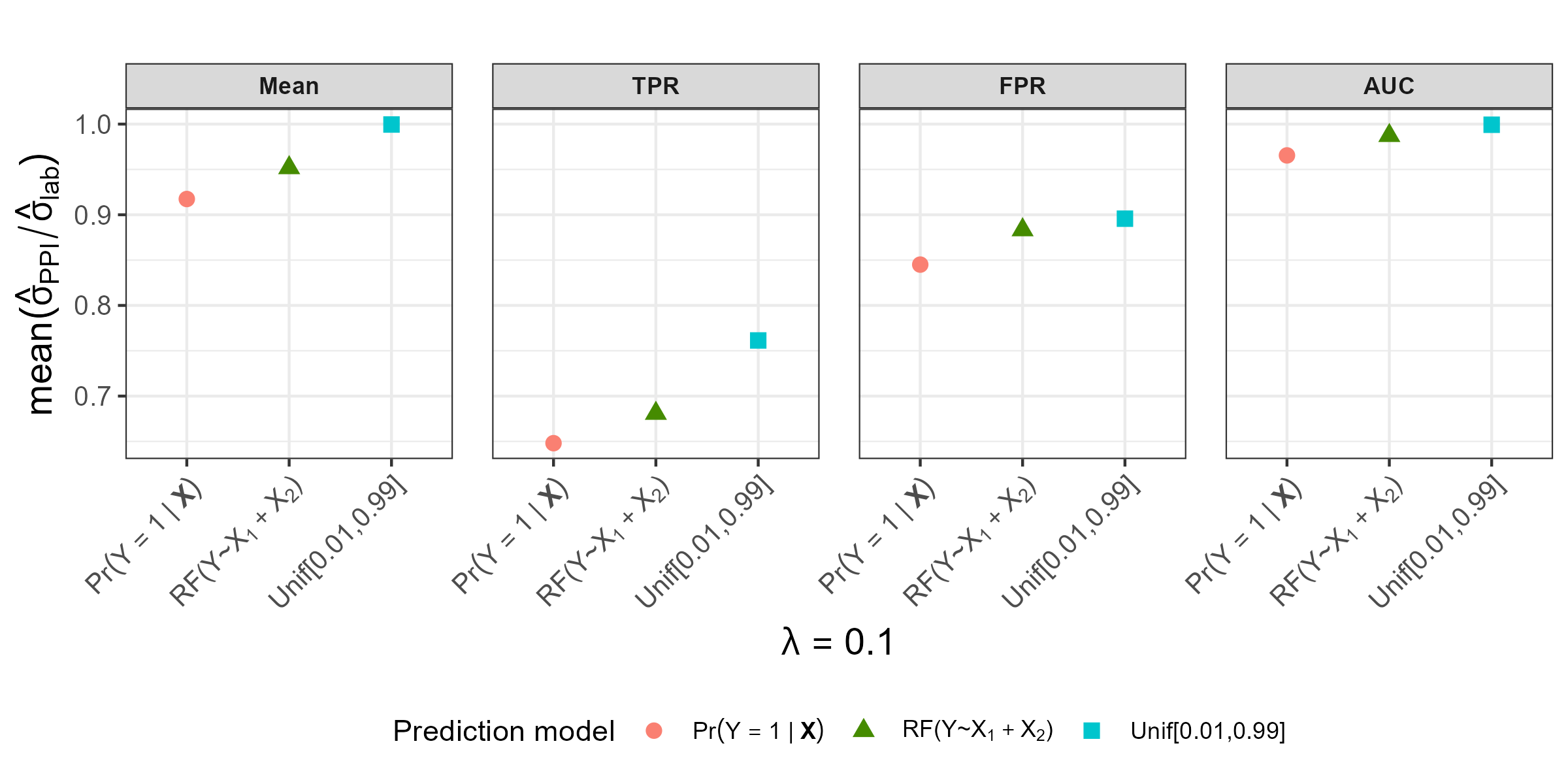}
    \caption{Ratio of estimated standard error for the PPI estimator and that of the estimator using only labeled data, averaged over 2500 replications, in the setting of Section \ref{sec:cordgp} with $n=1000$, $\lambda=0.1$, for mean, TPR, FPR, and AUC estimation. Shapes denote different prediction models $f$ used in PPI. The standard errors of PPI estimators are no greater that those of the estimators using only the labeled data. $\text{RF}(Y\sim X_1+X_2)$ does not use the feature involved in the data generation process, but still results in improvement.}
    \label{fig:cor_ratio}
\end{figure}
Figure \ref{fig:cor_ratio} compares the estimated standard errors of the PPI estimators to those of the estimators using only the labeled data for mean, TPR, FPR, and AUC estimation. TPR and FPR are evaluated at a threshold value of 0.6. Note that $\text{RF}(Y\sim X_1+X_2)$ does not contain the feature involved in the data generation process, but still results in improvement in variance across every estimand due to correlation with that feature.
\begin{figure}[ht]
    \centering
    \includegraphics[width=0.95\linewidth]{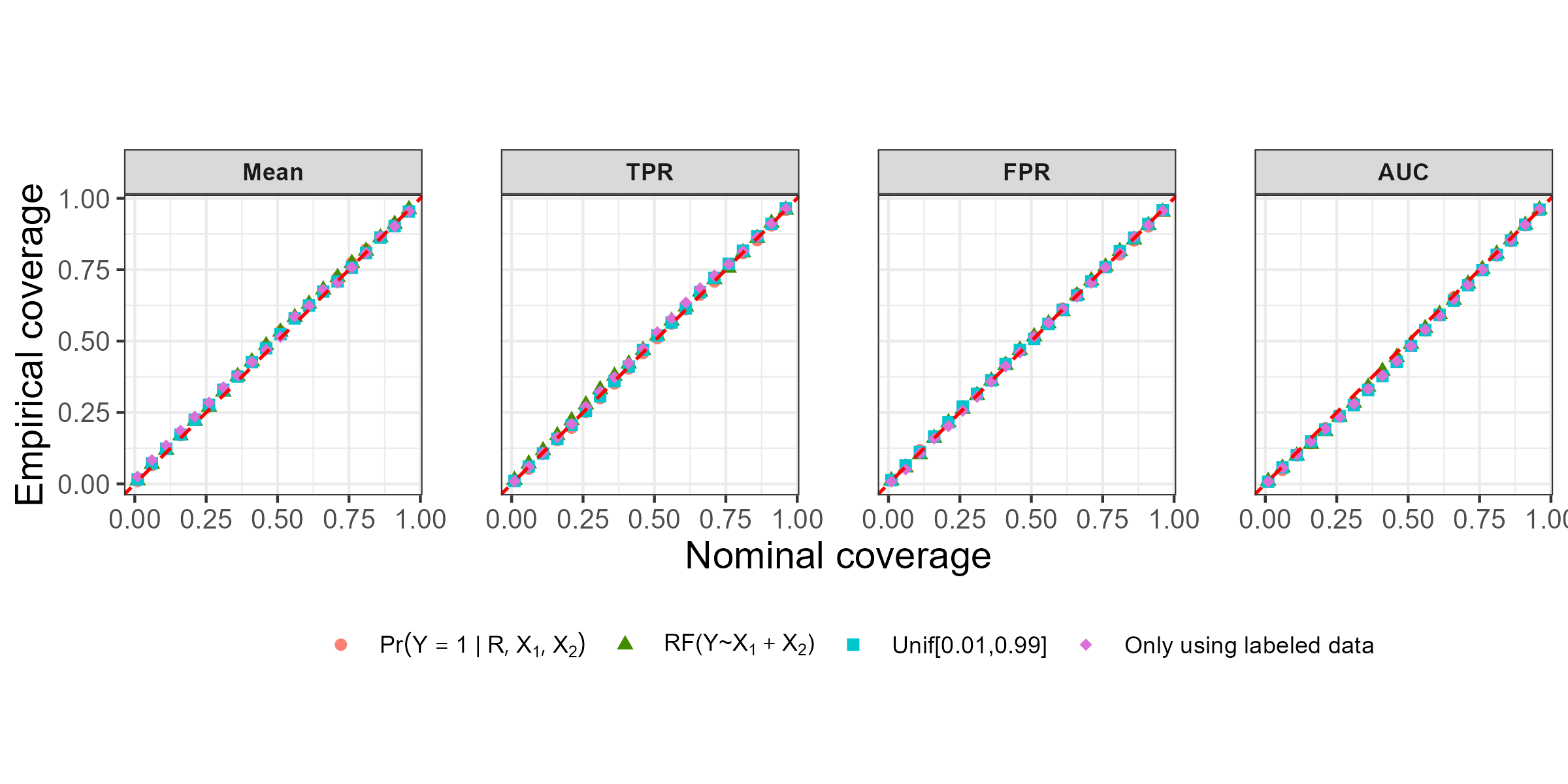}
    \caption{Empirical coverage versus nominal coverage, averaged over 2500 replications, in the setting of Section \ref{sec:cordgp} with $n=1000$, $\lambda=0.1$, for mean, TPR, FPR, and AUC estimation. The red dashed line is the diagonal line. The empirical coverages are close to the nominal coverages.}
    \label{fig:cor_cov}
\end{figure}
Figure \ref{fig:cor_cov} shows that the empirical coverages of the PPI estimators are close to the nominal coverages.
\subsubsection{Coverage results for setting $n=10,000$, $\lambda=0.1$, averaged over 10,000 replications}\label{sec:big_sim}
\begin{figure}[htbp]
    \centering
    \includegraphics[width=0.95\linewidth]{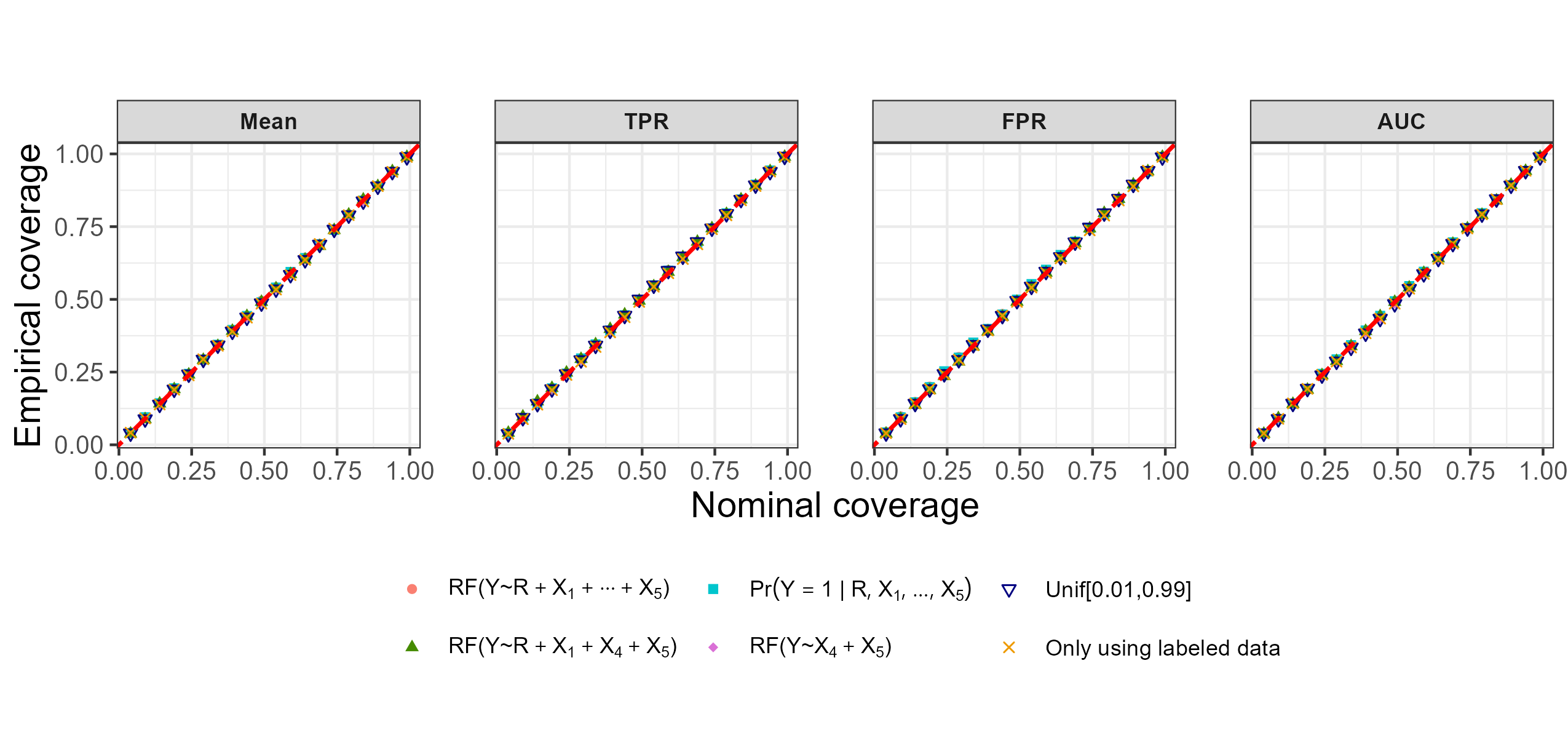}
    \caption{Empirical coverage versus nominal coverage, averaged over 10,000 replications, corresponding to Section \ref{sec:big_sim}, with $n=10,000$, $\lambda=0.1$, for TPR, FPR, and AUC estimation. The red dashed line is the diagonal line. The empirical coverages are close to the nominal coverages.}
    \label{fig:big_cov}
\end{figure}
Figure \ref{fig:big_cov} shows that the empirical coverages align very well with the nominal coverages, in the setting of Figure \ref{fig:coverage} in Section \ref{sec:simres_noncovshift}, but with a larger sample size and more simulated datasets.

\subsubsection{Bias results for Sections \ref{sec:sim_nocovshift} and \ref{sec:sim_covshift}}\label{sec:biasres}

Figures \ref{fig:bias_nocovshift} and \ref{fig:bias_covshift} show that the bias of PPI estimators in the settings of Sections \ref{sec:sim_nocovshift} and \ref{sec:sim_covshift} is close to 0.
\begin{figure}[ht]
    \centering
    \includegraphics[width=1\linewidth]{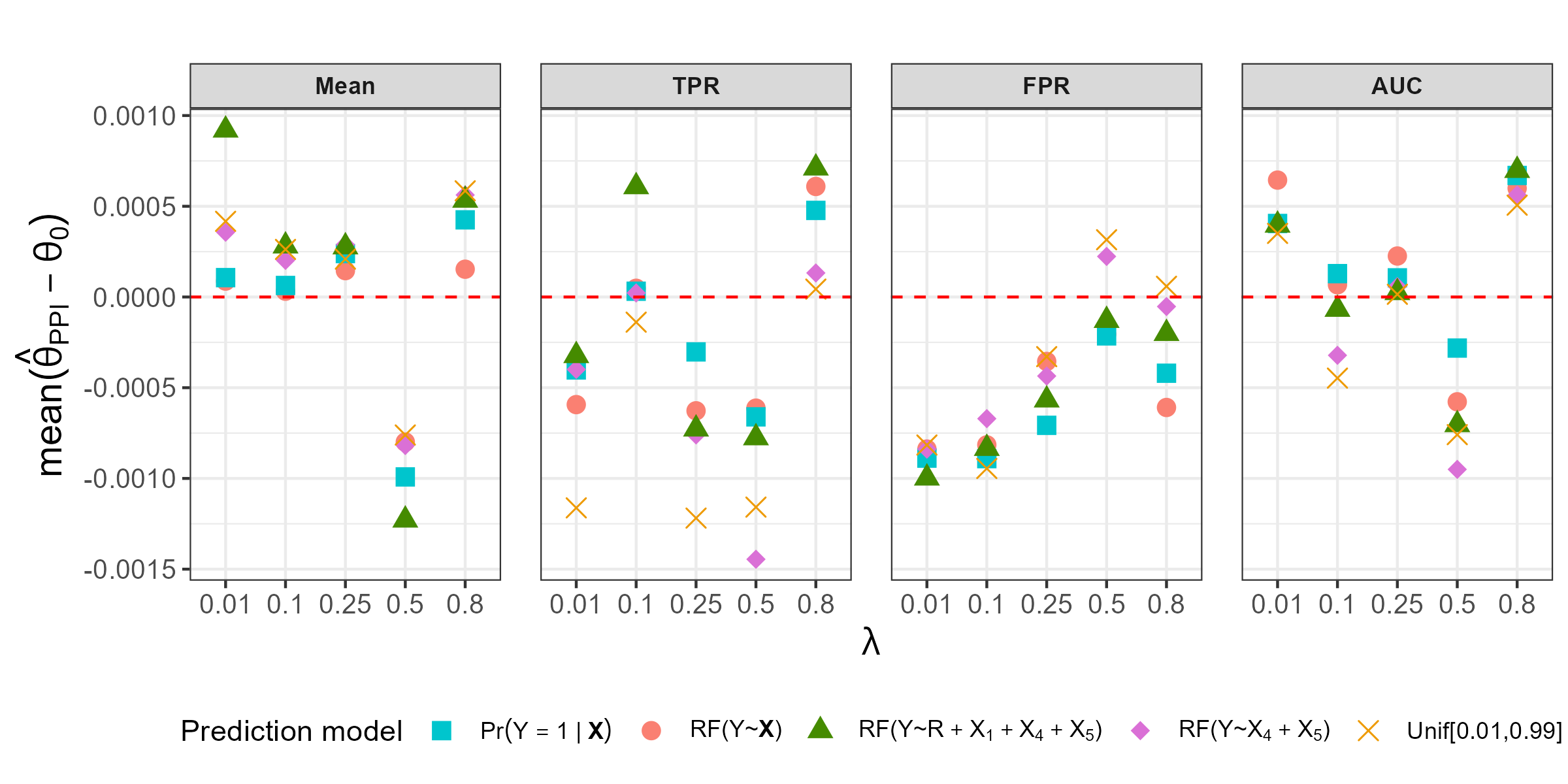}
    \caption{Bias of the PPI estimators, averaged over 2500 replications, in the setting of Section \ref{sec:sim_nocovshift} with $n=1000$, $\lambda\in \{0.01,0.1,0.25,0.5,0.8\}$, for mean, TPR, FPR, and AUC estimation. Shapes denote different prediction models $f$ used in PPI. In all cases, the bias is close to 0.}
    \label{fig:bias_nocovshift}
\end{figure}

\begin{figure}[ht]
    \centering
    \includegraphics[width=1\linewidth]{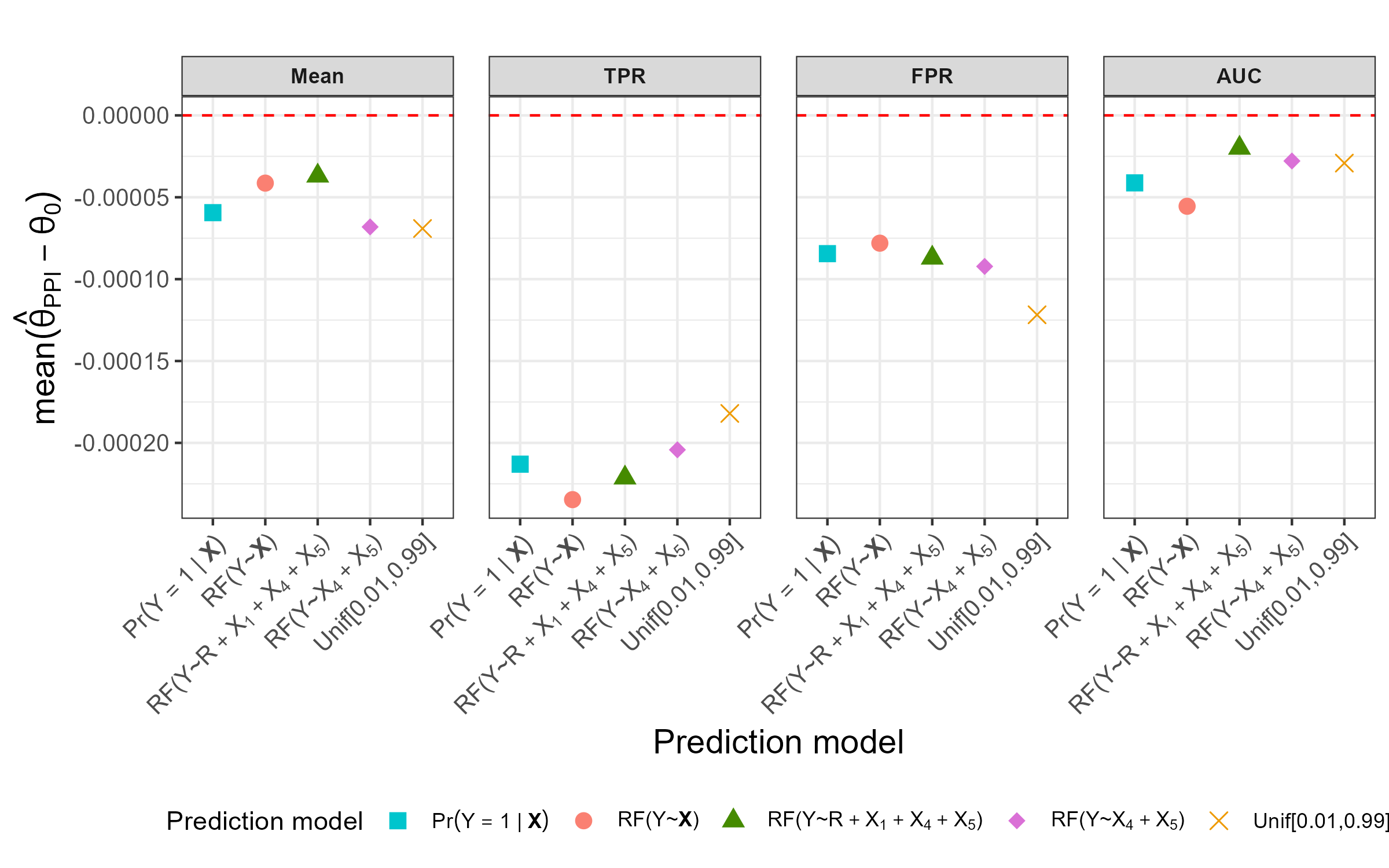}
    \caption{Bias of the PPI estimators, averaged over 2500 replications, in the setting of Section \ref{sec:sim_covshift} with $n+N=50,000$, for mean, TPR, FPR, and AUC estimation. Shapes denote different prediction models $f$ used in PPI. In all cases, the bias is close to 0.}
    \label{fig:bias_covshift}
\end{figure}